\documentclass[a4paper,USenglish,cleveref, autoref, thm-restate]{lipics-v2021}
\usepackage{xspace}
\usepackage[disable]{todonotes}

\usepackage[normalem]{ulem}
\usepackage{csquotes}
\usepackage{complexity}
\usepackage{tikz}
\usepackage{cite}
\usepackage{mathtools}

\usetikzlibrary{arrows.meta,patterns,ipe}

  \makeatletter
\pretocmd{\thmt@rst@storecounters}{\Hy@SaveLastskip}{}{}
\apptocmd{\thmt@rst@storecounters}{\Hy@RestoreLastskip}{}{}
\makeatother

\newboolean{long}
\usepackage{pifont}
\newcommand{\AppendixSymbol}{\ding{72}}
\setboolean{long}{false}  \ifthenelse{\boolean{long}}{
	\NewDocumentEnvironment{prooflater}{m}{\begin{proof}}{\end{proof}\ignorespacesafterend}
	\NewDocumentEnvironment{proofsketch}{o +b}{}{\ignorespacesafterend}
	\newcommand{\restateref}[1]{}
    \newcommand{\restaterefnoreference}[1]{}
	\NewDocumentEnvironment{statelater}{m}{}{}
	
	 	\NewDocumentCommand{\onlyShort}{+m}{}
	\NewDocumentCommand{\onlyLong}{+m}{#1}
	
}{
	\NewDocumentEnvironment{prooflater}{m +b}{ 		 		\expandafter\global\expandafter\def\csname#1\endcsname{\begin{proof}#2\end{proof}} 	}{\ignorespacesafterend}
	\NewDocumentEnvironment{proofsketch}{O{Proof sketch.}}{\begin{proof}[#1]}{\end{proof}\ignorespacesafterend}
	\usepackage{apptools}
	 	 	 	\newcommand{\restateref}[1]{[\IfAppendix{\hyperref[#1]{\AppendixSymbol{}}}{\hyperref[#1*]{\AppendixSymbol{}}}]}

    \newcommand{\restaterefnoreference}[1]{[\IfAppendix{\AppendixSymbol{}}{\AppendixSymbol{}}]}

	\NewDocumentEnvironment{statelater}{m +b}{ 		\expandafter\global\expandafter\def\csname#1\endcsname{#2} 	}{\ignorespacesafterend}

	 	\NewDocumentCommand{\onlyShort}{+m}{#1}
	\NewDocumentCommand{\onlyLong}{+m}{}
	
}
\let\oldrestatable\restatable
\def\restatable{\expandafter\oldrestatable}

 \makeatletter

\let\c@claim\relax
\makeatother
\theoremstyle{claimstyle}
\newtheorem{claim}{Claim}
 \AtBeginEnvironment{proof}{\setcounter{claim}{0}}

\title{
       Ramsey Meets Simultaneous Planarity:\\
Vertex Integrity and Pathwidth
 }
\titlerunning{Ramsey Meets Simultaneous Planarity:
Vertex Integrity and Pathwidth}

\author{Thomas Depian}{TU Wien, Austria}{tdepian@ac.tuwien.ac.at}{https://orcid.org/0009-0003-7498-6271}{}

\author{Simon D. Fink}{TU Wien, Austria}{sfink@ac.tuwien.ac.at}{https://orcid.org/0000-0002-2754-1195}{}

\author{Alexander Firbas}{TU Wien, Austria}{afirbas@ac.tuwien.ac.at}{https://orcid.org/0009-0007-2049-2144}{}

\author{Robert Ganian}{TU Wien, Austria}{rganian@gmail.com}{https://orcid.org/0000-0002-7762-8045}{}

\author{Matthias Pfretzschner}{University of Passau, Germany}{pfretzschner@fim.uni-passau.de}{https://orcid.org/0000-0002-5378-1694}{}

\author{Ignaz Rutter}{University of Passau, Germany}{rutter@fim.uni-passau.de}{https://orcid.org/0000-0002-3794-4406}{}

\Copyright{Thomas Depian, Simon D. Fink, Alexander Firbas, Robert Ganian, Matthias Pfretzschner, and Ignaz Rutter}
\authorrunning{T.\ Depian, S.\ D.\ Fink, A.\ Firbas, R.\ Ganian, M.\ Pfretzschner, and I.\ Rutter}

\funding{Thomas Depian, Simon D. Fink, Alexander Firbas, and Robert Ganian acknowledge support from the Vienna Science and Technology Fund (WWTF) [10.47379/ICT22029]. Alexander Firbas and Robert Ganian also acknowledge support from the Austrian Science Fund (FWF) [10.55776/Y1329]. Matthias Pfretzschner and Ignaz Rutter acknowledge support from the Deutsche Forschungsgemeinschaft (DFG, German Research Foundation) [541433306].}

 

\nolinenumbers

\theoremstyle{definition}

\crefname{kernel}{Kernel}{Kernels}
\Crefname{kernel}{Kernel}{Kernels}

\DeclareMathOperator{\Ima}{Im}

\newcommand{\mc}{\mathcal}
\newcommand{\tw}[1]{\mathrm{tw}(#1)}               \newcommand{\pw}[1]{\mathrm{pw}(#1)}                                                          \newcommand{\vc}[1]{\mathrm{vc}(#1)}                                     \newcommand{\vi}[1]{\mathrm{vi}(#1)}               
\newcommand{\dcup}{\mathop{\dot\cup}}

\newcommand{\att}{\text{att}}
\newcommand{\iso}{\ensuremath{\kappa}}
\newcommand{\phix}{\ensuremath{\varphi_{i,x}}}
\newcommand{\tphix}{\ensuremath{\tilde \varphi_{i,x}}}
\newcommand{\emb}{\ensuremath{\mc E}\xspace}
\newcommand{\embU}{\ensuremath{\emb_\cup}\xspace}
\newcommand{\embP}{\ensuremath{\tilde {\mc E}}\xspace}
\newcommand{\embPU}{\ensuremath{\embP_\cup}\xspace}

\newcommand{\upup}{\mathbin{\uparrow\uparrow}}

\newcommand{\Instance}{\ensuremath{(G_1, \dots, G_k)}}

\newcommand{\sefe}{\textsc{SEFE}\xspace}
\newcommand{\ssefe}{\textup{\textsc{S-SEFE}}\xspace}

\newcommand{\Ir}{\ensuremath{\mc I_\mathrm{r}}\xspace}
\newcommand{\Is}{\ensuremath{\mc I_\mathrm{s}}\xspace}

\newcommand{\expt}[2]{\ensuremath{\mathrm{exp}_2^{(#1)}(#2)}}

\newcommand{\union}{\ensuremath{G_\cup}\xspace}
\newcommand{\shared}{\ensuremath{G_\cap}\xspace}
\newcommand{\GU}{\union}

 \newcommand{\SEFE}{\ssefe}

 \EventEditors{John Q. Open and Joan R. Access}
\EventNoEds{2}
\EventLongTitle{42nd Conference on Very Important Topics (CVIT 2016)}
\EventShortTitle{CVIT 2016}
\EventAcronym{CVIT}
\EventYear{2016}
\EventDate{December 24--27, 2016}
\EventLocation{Little Whinging, United Kingdom}
\EventLogo{}
\SeriesVolume{42}
\ArticleNo{23}
 
\hideLIPIcs

\begin{document}

\maketitle

 \begin{abstract}
We study the following simultaneous embedding problem: Given $k$ graphs on the same vertex set with the same pairwise intersection, do they admit planar drawings that coincide on their shared graph? We establish fixed-parameter tractability with respect to the vertex integrity of the union graph (that is, the union of the input graphs) plus $k$, resolving an open question and generalizing the known algorithm parameterized by vertex cover number plus $k$. Our main technical contribution on the way towards this result is the extension of the recently introduced Ramsey pruning technique beyond ordering problems.
   Finally, we complement our tractability result with a lower bound that establishes \NP-hardness when both $k$ and the pathwidth of the union graph are bounded by constants, ruling out a generalization of our main result to path- or treewidth.
           \end{abstract}

% \setcounter{page}{0}
% \newpage

\section{Introduction}

The task of computing a plane drawing of a graph---or determining that none exists---has had a particularly broad impact on theoretical computer science. An important motivation for such drawings is that they give rise to more easily readable visualizations of graphs~\cite{Purchase00,PurchaseCA02}. Nevertheless, in many scenarios there is a set of graphs to visualize, rather than a single graph. Examples include dynamic networks that change over time, as well as networks obtained by applying different reconstruction methods to the same underlying data. In such settings, independently computing a planar drawing of each graph is insufficient: to allow a user to compare the drawings, the common parts of the graphs should be represented consistently. This task is captured by the \textsc{Simultaneous Embedding with Fixed Edges} (\sefe) problem: Given graphs $G_1,\dots,G_k$, decide whether they admit planar drawings $\Gamma_1,\dots,\Gamma_k$ such that every vertex and edge shared by two graphs is represented by the same point and the same curve, respectively, in both drawings.

An important and well-studied special case is the \emph{sunflower case}, where the intersection $G_i\cap G_j$ is the same graph $G_\cap$ for all distinct $i,j\in[k]$. We refer to this problem as \emph{sunflower \sefe}, or \ssefe, and call $G_\cap$ the \emph{shared graph}. Edges of the shared graph are \emph{shared}, while every other edge is \emph{exclusive} to one input graph. Without loss of generality, all input graphs have the same vertex set $V$, since exclusive vertices can be added as isolated vertices to the other graphs. A sunflower instance can thus be represented by its \emph{union graph} $G_\cup$, whose edge set is partitioned into the shared edges $E_\cap$ and the exclusive edges $E_1,\dots,E_k$. The task is to draw $G_\cup$ so that, for each $i\in[k]$, the drawing restricted to $E_\cap\cup E_i$ is planar. Throughout this paper, we consider this sunflower setting.

The computational complexity of simultaneous planarity has received considerable attention over the last two decades. The problem is \NP-complete already for three input graphs~\cite{sefe_in_np}, even in the sunflower case~\cite{Schaefer12} and with the additional restriction that the shared graph is a tree~\cite{AngeliniLN15}. For two input graphs, every instance is a sunflower instance and the complexity of the problem remains open. Nevertheless, a substantial body of work has identified increasingly general tractable special cases. In particular, the problem can be solved in linear time if the embedding of the shared graph is fixed~\cite{AngeliniBFJKPR15} or if the shared graph is biconnected~\cite{DBLP:journals/jda/AngeliniBFPR12, DBLP:journals/jgaa/HaeuplerJL13}. More recently, polynomial-time algorithms have been established for two input graphs with a connected shared graph~\cite{FulekT22,BlasiusFR23,FinkR25}. We refer to the survey by Rutter~\cite{Rutter20} for a more detailed overview.

A useful perspective on these results is provided by the combinatorial formulation of J\"unger and Schulz~\cite{DBLP:journals/jgaa/JungerS09}: a sunflower instance admits a simultaneous drawing if and only if its input graphs admit planar embeddings that induce the same embedding on the shared graph. Here, an embedding comprises the cyclic order of incident edges around each vertex, called the \emph{rotation system}, together with the relative positions of connected components. In particular, \ssefe asks whether the shared graph has an embedding that can be independently extended to every input graph. This connects simultaneous planarity to the study of partial drawing extension, for which several fixed-parameter algorithms have been developed in recent years~\cite{EibenGHKN20,GanianHKPV21,BhoreGKMN23,CGK26}. However, those algorithms typically parameterize by the part of the graph that is still to be drawn, whereas in simultaneous planarity the embedding of the shared graph is itself unknown; moreover, all of the listed examples impose connectivity constraints on the input graph.

In this work, we target the parameterized
complexity of \ssefe, with particular emphasis on the structure of the union graph. This research direction was initiated by Fink, Pfretzschner, and Rutter~\cite{sefe_fpt_v1}, who established fixed-parameter tractability with respect to the number of input graphs $k$ plus either the \emph{vertex cover number} or the \emph{feedback edge number} of $G_\cup$. Recently, it was shown that the dependence on $k$ can be avoided for feedback edge number, but is necessary for vertex cover number~\cite{DFFGPR25}: \ssefe remains \NP-complete even when $G_\cup$ has vertex cover number at most~$15$. These results leave open how far the known fixed-parameter tractability can be extended when $k$ is included in the parameterization; as a concrete target, the aforementioned article~\cite{DFFGPR25} explicitly asked whether it can be extended to the \emph{vertex integrity} of the union graph.

The vertex integrity $\vi{G}$ of a graph $G$ measures, roughly speaking, how many vertex deletions are required to decompose $G$ into small connected components; it is the minimum, over all modulators $X\subseteq V(G)$, of $|X|$ plus the size of a largest connected component of $G-X$. The vertex integrity has proved useful in the design of parameterized algorithms for a variety of challenging problems~\cite{dvovrak2017solving,bodlaender2020subgraph,ganian2021structural,ganian2021structuralpaths,dvovrak2021complexity,GimaHKKO22,LampisM24,hanaka2024parameterized}\footnote{Some older works use the term \emph{fracture number}, which is parametrically equivalent to vertex integrity.} and is also a natural choice here since $\vi{G}\leq\vc{G}+1$, but graphs of bounded vertex integrity can have arbitrarily large vertex cover number. 
 
\smallskip
\noindent\textbf{Contributions.}\quad
Our main result settles the above open question and extends the known fixed-parameter tractability of \ssefe beyond vertex cover number:

\begin{restatable}{theorem}{vertexintegrityfpt}
\label{theorem:vertex_integrity_fpt}
\SEFE is fixed-parameter tractable when parameterized by the vertex integrity of the union graph plus the number $k$ of input graphs.
\end{restatable}

 The main challenge behind \cref{theorem:vertex_integrity_fpt} is to identify which parts of a simultaneous embedding can be captured by a bounded amount of combinatorial information. A standard approach for vertex integrity groups the components outside a modulator according to their isomorphism type, including their attachments to the modulator, and retains a bounded number of representatives of each type to produce a smaller graph (a \emph{kernel}) $G_\cup'$. For \ssefe, these types must also record which edges belong to which input graph; the goal would then be to show that a solution for $G_\cup'$ can be lifted to $G_\cup$\footnote{In particular, we would like an equivalence between $G_\cup'$ and $G_\cup$---however, the other direction follows trivially by the former being a subgraph of the latter.}. 
 
 Unfortunately, a simultaneous embedding of the retained representatives need not allow for the insertion of another isomorphic component: the existing copies can be interleaved around attachment vertices in a way that prevents such an extension. A concrete example of this behavior is provided in Figure~\ref{fig:noextension}, and this rules out lifting standard kernelization arguments from the vertex cover and feedback edge numbers~\cite{sefe_fpt_v1,DFFGPR25} to vertex integrity.

\begin{figure}
  \centering
  \hfill
  \begin{subfigure}{0.1\textwidth}
    \centering
    \includegraphics[page=1]{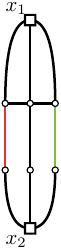}
    \caption{}
  \end{subfigure}
  \hfill
  \begin{subfigure}{0.4\textwidth}
    \centering
    \includegraphics[page=2]{figures/counterexample}
    \caption{}
  \end{subfigure}
  \begin{subfigure}{0.45\textwidth}
  	\centering
  	\includegraphics[page=3]{figures/counterexample}
  	\caption{}
  \end{subfigure}
  \hfill~
  \caption{
    \textbf{(a)} A union graph $H_\cup$ on two input graphs. Shared edges are drawn in black.
    \textbf{(b)} The union graph $\GU$ obtained by joining four copies of $H_\cup$ at $\{x_1, x_2\}$, together with a corresponding simultaneous embedding; the individual copies are highlighted in colors.
    Note that the depicted solution cannot be extended by adding any additional copies of $H_\cup$ without altering the given solution.
    A similar ``counterexample'' can be constructed for any number of copies of $H_\cup$ by adapting $\GU$ accordingly. 
    \textbf{(c)} A different simultaneous embedding of $\GU$ that can be extended by arbitrarily many additional copies of $H_\cup$ without altering the given solution.
   }
  \label{fig:noextension}
\end{figure}

A promising way to overcome this obstacle is the \emph{Ramsey pruning} technique introduced by Depian, Fink, Ganian, and Surianarayanan~\cite{DFGS.LLR.2025} for computing stack and queue layouts. At a high level, it starts by preprocessing an instance until it consists solely of a modulator and a large number of pairwise-isomorphic and pairwise-independent \emph{pieces}.  Its core step then uses Ramsey theory to extract a sufficiently large subset of pieces with pairwise-uniform interactions in a solution (thus avoiding non-extendable solutions such as the one illustrated in Figure~\ref{fig:noextension}(b), instead guaranteeing solutions such as the one in Figure~\ref{fig:noextension}(c)), and then exploits this uniformity to reintroduce deleted pieces. Ramsey pruning has by now also been extended to more general linear layout problems~\cite{DFFGNS26} as well as to computing the \emph{twin-width} of an input graph~\cite{GanianRocton26} w.r.t.\ the vertex integrity. In all of these settings, the relevant interactions are primarily captured by a total order: either of the vertices (in layout problems), or of contractions (in the case of twin-width). The main technical difficulty is that there is no immediate way of encoding two-dimensional simultaneous drawings as a total order: it is necessary to capture the pairwise interactions between pieces in a way which guarantees both (i) a bounded number of interaction types and (ii) extendability. 
 
Our starting point is to encode the drawings through the cyclic orders of edges around vertices. For connected input graphs, these rotation systems determine the individual embeddings. The crucial ingredient that allows us to use this representation is a local characterization of planar rotation systems (\Cref{lm:planarityCharacterization}), inspired by the cycle-and-bridge perspective underlying the planarity test of Auslander and Parter~\cite{Auslander1961}. Roughly speaking, a rotation system is planar if and only if, for every simple cycle, all fixtures of each bridge lie on the same side of the cycle. A violation can therefore be witnessed by a cycle together with a path in one of its bridges. In our setting, such a witness involves only a bounded number of pieces outside the modulator. This provides the link between uniformity of cyclic orders and preservation of planarity that is needed for Ramsey pruning, and opens the door towards applying the latter technique to other drawing and representation problems.

A caveat of this argument is that it requires each input graph to be connected. To deal with this, we revisit the connectivity reduction of Bl\"asius and Rutter~\cite{BlasiusR15} to show that suitable gadgets can make every input graph connected while keeping the two parameters bounded.
This construction does not, however, guarantee connectivity of the shared graph---to deal with this obstacle, we ensure that our reconstruction procedure is sufficiently rigid to also handle disconnected shared graphs.

Finally, a natural question is how far Theorem~\ref{theorem:vertex_integrity_fpt} can be extended towards less restrictive structural parameters. We complement our algorithm with a negative answer for pathwidth:
\begin{restatable}\restaterefnoreference{theorem:pathwidth_hardness}{theorem}{pathwidthhardness}
\label{theorem:pathwidth_hardness}
\SEFE remains \NP-hard even when there are $3$ input graphs and the union graph has pathwidth at most $5$.
\end{restatable}
This not only effectively rules out fixed-parameter tractability but also \XP-membership with respect to $k+\pw{G_\cup}$, and hence also $k+\tw{G_\cup}$, resolving a second open question of previous work~\cite{DFFGPR25}.

% \onlyShort{
% \smallskip
% \noindent
% \emph{Full details for statements marked by \AppendixSymbol\ can be found in the appendix.}
% }
 
\section{Preliminaries}
\label{section:prelims}  

We assume familiarity with basic concepts in parameterized complexity theory, including the notions of \emph{fixed-parameter tractability} and para\NP-hardness~\cite{Cygan2015}.
For two functions $f_1 \colon A \to B$ and $f_2 \colon C \to D$ with $f_1(x) = f_2(x)$ for all $x \in A \cap C$, we define $f_1 \cup f_2 \colon A \cup C \to B \cup D$ with
$$(f_1 \cup f_2)(x) = \begin{cases}
  f_1(x), & \text{if $x \in A$}\\
  f_2(x), & \text{otherwise}
\end{cases}$$

Moreover, we sometimes express bounds on functions using the Knuth notation $\uparrow\uparrow$ where for an integer $z$, $2\uparrow\uparrow z$ represents an exponential tower of $2$'s of height $z$.
For a positive integer $L$, we write $\mathbb Z_L=\mathbb Z/L\mathbb Z$ for the integers modulo $L$.

\subparagraph{Graphs, Drawings, and Embeddings.}
We assume familiarity with standard graph terminology~\cite{Die.GT4.2012}.
       The \emph{vertex integrity $\vi G$} of a graph $G$ is the smallest integer $\psi$ such that there exists a set $X \subseteq V(G)$ for which the sum of $|X|$ plus the maximum size of a connected component of $G - X$ is at most $\psi$.
 We refer to $X$ as a \emph{witness} of $\vi G$.

A \emph{rotation system} of $G$ assigns to each vertex a clockwise circular order of its incident edges.
A \emph{drawing} $\Gamma$ of $G$ assigns each vertex $v \in V(G)$ to a point in the plane and each edge $uv \in E(G)$ to a Jordan arc between $\Gamma(u)$ and $\Gamma(v)$.
The drawing $\Gamma$ is \emph{planar} if no two such arcs intersect at an interior point.
A planar drawing partitions the plane into regions called \emph{faces}.
Moreover, it induces a rotation system of $G$ and \emph{relative positions}, i.e., an assignment of each connected component to a face of every other connected component.
A \emph{planar embedding} $\mathcal E$ of $G$ is an equivalence class of planar drawings with the same rotation system and relative positions.
For a vertex $v$ of $G$, we let $\mathcal E(v)$ denote the circular order $\mathcal E$ induces for the edges incident to $v$.
We mostly consider planar embeddings of connected graphs, which are entirely defined by a rotation system.
A rotation system of a connected graph $G$ is \emph{planar} if it is the rotation system induced by some planar drawing of~$G$.

\subparagraph{A Characterization for Planar Rotation Systems.} 
Let $G$ be a connected graph and let $C$ be a simple cycle of $G$.
A \emph{bridge} of $C$ is either a chord or a connected component of $G - C$ together with the edges connecting it to $C$.
 Let $B$ be a bridge of $C$.
A vertex of $C$ is a \emph{fixture vertex} of $B$ if it is incident to an edge of $B$.
An edge of $B$ is a \emph{fixture edge} if it is incident to a fixture vertex of $B$.
A \emph{fixture} of $B$ is a pair $(v, e)$, where $e$ is a fixture edge of $B$ and $v$ is a fixture vertex of $B$ incident to~$e$.
If $B$ consists of a single edge, note that its edge $e$ belongs to two fixtures of $B$.
We assume that $C$ is oriented and thus a rotation system $\mathcal E$ of $G$ assigns each fixture $a =(v, e)$ of $B$ to one side of $C$: If $e$ appears after the incoming and before the outgoing edge of $C$ in the clockwise circular order $\mathcal{E}(v)$ we say that $a$ lies \emph{inside of} $C$, otherwise it lies \emph{outside of} $C$.

We use the following characterization for planarity of rotation systems based on fixtures of cycles.
We remark that this characterization is similar to the characterization of planar graphs used by Lov{\'{a}}sz et al.~\cite{LPS1997} and in the planarity test by Auslander and Parter~\cite{Auslander1961}.
Moreover, Mohar and Thomassen~\cite{MT2001} prove the statement for more general surfaces.
For completeness, we give a self-contained elementary proof in the appendix.

\begin{restatable}\restateref{lm:planarityCharacterization}{lemma}{lemmaPlanarityCharacterization}
\label{lm:planarityCharacterization}
    Let $G$ be a connected graph and let $\mc E$ be a rotation system of $G$.
    Then $\mc E$ is planar if and only if, for every simple cycle $C$ of $G$, every bridge $B$ of $C$, and every pair of fixtures $a_1 = (v_1, e_1), a_2 =(v_2, e_2)$ of $B$, $a_1$ lies inside $C$ in $\mc E$ if and only if $a_2$ does.
\end{restatable}
\begin{prooflater}{plemmaPlanarityCharacterization}
It is clear that the condition is necessary; we thus only show that it is also sufficient.
Let $\mc E$ be a rotation system of $G$ that satisfies the condition of the lemma.
Note that $\mc E$ is planar if and only if (1) $\mc E$ induces a planar embedding for each block of $G$ and (2) for each cutvertex $v$ of $G$ and every pair of distinct split components $S_1, S_2$ of $v$, the edges of $S_1$ incident to $v$ are consecutive with respect to the edges of $S_2$ incident to $v$.

We first show that $\mc E$ satisfies statement (2).
If $S_1$ or $S_2$ only has a single edge incident to~$v$, the statement holds trivially.
Otherwise, assume for the sake of contradiction that $v$ has incident edges $f_1, f_2$ that belong to $S_1$ and edges $g_1, g_2$ that belong to $S_2$, such that these four edges appear in the clockwise order $f_1, g_1, f_2, g_2$ in $\mc E(v)$.
Because $f_1$ and $f_2$ belong to~$S_1$, there exists a simple cycle $C$ in $S_1$ that contains $v$, $f_1$, and $f_2$.
Moreover, since $g_1$ and $g_2$ both belong to $S_2$, they belong to the same bridge $B$ of $C$.
But the ordering $f_1, g_1, f_2, g_2$ places the fixtures $(v, g_1)$ and $(v, g_2)$ of $B$ on different sides of $C$, a contradiction.

To show that statement (1) holds, we thus now assume that $G$ is biconnected.
We prove the statement by induction on the number of edges of $G$.
We distinguish two cases, based on whether there exists a simple cycle of $G$ with more than one bridge or not.

\begin{figure}
  \centering
  \begin{subfigure}{0.3\textwidth}
    \centering
    \includegraphics[page=1]{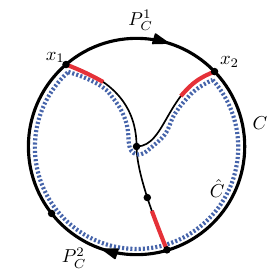}
    \caption{}
    \label{fig:treeBridge}
  \end{subfigure}
  \hfill
  \begin{subfigure}{0.3\textwidth}
    \centering
    \includegraphics[page=2]{figures/planarityCharacterization}
    \caption{}
    \label{fig:alternatingAttachments}
  \end{subfigure}
  \hfill
  \begin{subfigure}{0.3\textwidth}
    \centering
    \includegraphics[page=3]{figures/planarityCharacterization}
    \caption{}
    \label{fig:alternatingAttachments2}
  \end{subfigure}
  \caption{
  Illustration of the proof of \Cref{lm:planarityCharacterization}.
  \textbf{(a)} A cycle $C$ with a single bridge. Because this bridge is more than a path, there exists a second cycle $\hat C$ with more than one bridge. 
  \textbf{(b)-(c)} Two examples for a cycle $C$ with two bridges whose fixtures (red) alternate inside $C$. In both cases, there exists a cycle $\hat C$ that has a bridge with fixtures (green) that lie on different sides of $\hat C$.
  }
\end{figure}

\proofsubparagraph{Case 1:} Every simple cycle of $G$ has at most one bridge. 
We show that $G$ has very restricted structure.
\begin{claim}
    Let $G$ be a biconnected graph with at least two edges such that every simple cycle of $G$ has at most one bridge.
    Then $G$ is either a simple cycle or a simple cycle whose only bridge is a path.
\end{claim}
\begin{claimproof}
    Because $G$ is biconnected and has at least two edges, $G$ contains at least one simple cycle $C$.
    If $C$ has no bridges, then we have $G = C$.
    Otherwise, let $B$ denote the unique bridge of $C$.
    Due to biconnectivity, $B$ has at least two distinct fixture vertices.
    Choose two fixture vertices $x_1 \neq x_2$ of $B$ such that no other fixture vertex of $B$ lies between $x_1$ and $x_2$ along the orientation of $C$; see \Cref{fig:treeBridge}.
    Let $P^1_C$ be the $x_1x_2$-path that follows the orientation of $C$ and let $P^2_C$ be the other $x_1x_2$-path of $C$.
    Moreover, let $P_B$ be an arbitrary $x_1x_2$-path in $B$ that is disjoint from $C$, except for its endpoints.
    Consider the cycle $\hat C = P^2_C \cup P_B$.
    By the choice of $x_1, x_2$ and $P^1_C$, the interior vertices of $P^1_C$ have degree 2 in $G$.
    Hence $P^1_C$ constitutes an entire bridge of $\hat C$.
    Since $\hat C$ has no other bridges by assumption, we have $G = \hat C \cup P^1_C = P^2_C \cup P_B \cup P^1_C = C \cup P_B$.
    Because $G = C \cup B$, $B$ thus only consists of the path $P_B$, which concludes the proof.
\end{claimproof}

If $G$ has fewer than two edges or $G$ is a simple cycle, $\mc E$ is trivially planar.
By the claim, $G$ otherwise contains a simple cycle $C$ whose only bridge $B$ is a path.
Note that $\mc E$ trivially induces planar embeddings for $C$ and $B$.
Since $\mc E$ additionally places the two fixtures of $B$ on the same side of $C$ by assumption, we thus have that $\mc E$ is a planar embedding of $G = C \cup B$.

\proofsubparagraph{Case 2:} There exists a simple cycle $C$ of $G$ with at least two bridges.
Recall that, by assumption, the fixtures of each bridge of $C$ lie on the same side of $C$ in $\mc E$.
Let $\mc I$ and $\mc J$ denote the bridges whose fixtures lie inside and outside of $C$, respectively.
If $\mc I$ and $\mc J$ are both non-empty, we obtain by induction that $\mc E$ induces planar embeddings of the smaller subgraphs $C \cup \bigcup_{B \in \mc I} B$ and $C \cup \bigcup_{B \in \mc J} B$ and thus $\mc E$ also induces a planar embedding of $G$.
In the following, we thus assume without loss of generality that $\mc J = \emptyset$, i.e., the fixtures of all bridges of $C$ lie inside of $C$ in $\mc E$.

Note that the chosen orientation of $C$ defines a circular order of all fixtures inside $C$: fixtures at distinct vertices are ordered by the orientation of $C$, fixtures at a common vertex $v$ are ordered by $\mc E(v)$, from the incoming edge of $C$ towards the outgoing edge.
We first show that the fixtures of one of the bridges appear consecutively in this circular order.
\begin{claim}
    There exists a bridge $B \in \mc I$ such that the fixtures of $B$ appear consecutively inside $C$.
\end{claim}
\begin{claimproof}
    We first show that the fixtures of two distinct bridges of $\mc I$ do not alternate around $C$.
    Assume for the sake of contradiction that there exist fixtures $a_1 = (v_1, e_1), a_2 = (v_2, e_2)$ of a bridge $B \in \mc I$ and fixtures $a'_1 = (v'_1, e'_1), a'_2 = (v'_2, e'_2)$ of a different bridge $B' \in \mc I$ that appear in the circular order $a_1, a'_1, a_2, a'_2$ inside $C$; see \Cref{fig:alternatingAttachments,fig:alternatingAttachments2}.

    First assume $v_1 \neq v_2$.
    Let $P_B$ denote an arbitrary $v_2v_1$-path of $B$ that starts with $a_2$, ends with $a_1$, and is otherwise disjoint from $C$.
    Let $\hat C$ denote the oriented cycle obtained by concatenating $P_B$ and the $v_1v_2$-subpath $P^1_C$ of $C$ that follows the orientation of $C$; see \Cref{fig:alternatingAttachments}.
    
    Because $e'_1$ and $e'_2$ both belong to $B'$, they also belong to a common bridge $\hat B$ of $\hat C$.
    Moreover, since $a_1'$ lies between $a_1$ and $a_2$ in the circular order inside $C$, $v_1'$ belongs to $P^1_C$ and thus $a_1'$ is a fixture of $\hat B$ that lies inside $\hat C$ in $\mc E$.
    We distinguish whether $v_2' \in \{v_1, v_2\}$ or not.
    
    If $v_2' \in \{v_1, v_2\}$, then $v_2'$ belongs to $\hat C$ and thus $a_2'$ is also a fixture of $\hat B$. 
    Since $a_2'$ lies between $a_2$ and $a_1$ in the circular order inside $C$, $a_2'$ lies outside of $\hat C$.
    Hence the fixtures $a_1'$ and $a_2'$ lie on different sides of $\hat C$ in $\mc E$, a contradiction.
    
    If $v_2' \notin \{v_1, v_2\}$, then the second $v_1v_2$-path $P^2_C \neq P^1_C$ of $C$ also belongs to $\hat B$.
    Let $\hat e$ denote the edge of $P^2_C$ incident to $v_2$ and observe that $\hat a = (v_2, \hat e)$ is a fixture of $\hat B$.
    Because $a_2$ lies inside $C$, $\hat a$ lies outside of $\hat C$ and we again obtain a contradiction.

    Finally, consider the case $v_1 = v_2$.
    If $v_1' \neq v_2'$, we can swap the labels of $B$ and $B'$ and their corresponding fixtures (which does not affect the circular order of the four fixtures inside $C$) and we again obtain the other case.
    We thus assume $v_1' = v_2'$.
    Since fixtures at a common vertex appear consecutively in the circular order, we therefore have $v_1 = v_1' = v_2 = v_2'$.
    Pick $\hat C$ as an arbitrary simple cycle of $B$ that contains $a_1$ and $a_2$ and is otherwise disjoint from $C$.
    Then $a_1'$ and $a_2'$ are fixtures of a bridge of $\hat C$ that lie on different sides of $\hat C$; see \Cref{fig:alternatingAttachments2}.
    We thus again obtain a contradiction.

    Hence the fixtures of two distinct bridges do not alternate inside $C$.
    Therefore, there exists a bridge $B$ whose fixtures appear consecutively inside $C$.
\end{claimproof}

  Let $B$ be a bridge of $\mc I$ whose fixtures appear consecutively inside $C$.
Because $|\mc I| \geq 2$, we obtain by induction that $\mc E$ induces planar embeddings $\mc E_1, \mc E_2$ of $C \cup B$ and $C \cup \bigcup_{B' \in (\mc I \setminus \{B\})} B'$, respectively.
Since the fixtures of $B$ are consecutive inside $C$, $\mc E$ places them in a single face  of $\mc E_2$.
Hence $\mc E$ is a planar embedding of $G$. 
\end{prooflater}

\subparagraph{Angles and Angle-Assignments.}

We will later have to merge multiple circular orderings into one, given specific information on how these orderings should interleave.
We formalize this using \emph{angles} and \emph{angle-assignments}.

Let $A, B$ be disjoint sets with $|A| \geq 2$ and let $\sigma$ be a circular ordering of $A$.
An \emph{angle} of~$\sigma$ is an ordered pair $\alpha = (a_1, a_2)$ of distinct elements in~$A$.
Note that $\alpha$ uniquely defines a contiguous linear subsequence $\sigma[\alpha]$ of $\sigma$ that starts with $a_1$ and ends with $a_2$.
If $\sigma[\alpha]$ only consists of $a_1$ and $a_2$ (i.e., $a_2$ directly succeeds $a_1$ in $\sigma$), we refer to $\alpha$ as a \emph{direct angle}.
 We say that an element $a \in A \setminus \{a_1, a_2\}$ \emph{lies inside} $\alpha$ if $a$ is contained in $\sigma[\alpha]$.
Moreover, an angle $\alpha'$ of $\sigma$ \emph{lies inside} $\alpha$ if $\sigma[\alpha']$ is a subsequence of $\sigma[\alpha]$.
 
\begin{figure}
  \centering
  \begin{subfigure}{0.3\textwidth}
    \centering
    \includegraphics[page=1,scale=0.9]{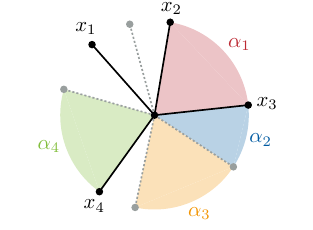}
    \caption{}
    \label{fig:anglesA}
  \end{subfigure}
  \hfill
  \begin{subfigure}{0.3\textwidth}
    \centering
    \includegraphics[page=2,scale=0.9]{figures/angles}
    \caption{}
    \label{fig:anglesB}
  \end{subfigure}
  \hfill
  \begin{subfigure}{0.3\textwidth}
    \centering
    \includegraphics[page=3,scale=0.9]{figures/angles}
    \caption{}
    \label{fig:anglesC}
  \end{subfigure}
  \caption{\textbf{(a)} A circular order $\sigma_1$ with four direct angles $\alpha_1, \alpha_2, \alpha_3, \alpha_4$ highlighted.
    \textbf{(b)} A circular order $\sigma_2$ that shares elements $X = \{x_1, x_2, x_3, x_4\}$ with $\sigma_1$. 
    A direct angle-assignment $\varphi$ from the elements of $\sigma_2$ not in $X$ to $\sigma_1$ is illustrated using colors.
         \textbf{(c)} The unique corresponding circular ordering that extends $\sigma_1$ and $\sigma_2$ and satisfies $\varphi$.
  }
  \label{fig:angles}
\end{figure}

An \emph{angle-assignment} from $B$ to $\sigma$ is a function $\varphi\colon B \to \{(a_1,a_2) \in A^2 \mid a_1\neq a_2\}$ that assigns each element of $B$ to an angle of $\sigma$.
If the image of $\varphi$ only consists of direct angles of~$\sigma$, we say that $\varphi$ is a \emph{direct angle-assignment}.
An angle-assignment $\varphi'$ from $B$ to a circular ordering $\hat \sigma \supseteq \sigma$ is a \emph{refinement of $\varphi$} if $\varphi'(b)$ lies inside $\varphi(b)$ in $\hat \sigma$ for all $b \in B$.
The refinement is \emph{tidy} if additionally $\varphi(b_1) = \varphi(b_2)$ implies $\varphi'(b_1) = \varphi'(b_2)$ for all $b_1, b_2 \in B$.
 A circular ordering~$\hat \sigma \supseteq \sigma$ that contains $B$ \emph{satisfies} $\varphi$ if every $b \in B$ lies inside $\varphi(b)$ in $\hat \sigma$; see \Cref{fig:angles}.
Conversely, every circular ordering $\hat \sigma \supseteq \sigma$ that contains $B$ \emph{induces} a unique direct angle-assignment from $B$ to $\sigma$ that it satisfies.
 
We show that, for two circular orderings sharing an element and a direct angle-assignment from one to the other, there is at most one extension that satisfies the~angle-assignment.
 
\begin{restatable}\restateref{lm:merges}{lemma}{lemmaMerges}
    \label{lm:merges}
    Let $X, A_1, A_2$ be pairwise disjoint sets with $|X| \geq 1$ and $|X \cup A_1| \geq 2$.
    Let $\sigma_1$ and $\sigma_2$ be circular orderings of $X \cup A_1$ and $X \cup A_2$, respectively,
         and let $\varphi$ be a direct angle-assignment from $A_2$ to $\sigma_1$.
    Then there exists at most one circular ordering of $X \cup A_1 \cup A_2$ that extends $\sigma_1$ and $\sigma_2$ and that satisfies $\varphi$.
                                                  \end{restatable}
\begin{prooflater}{plemmaMerges}
    Let $\sigma_{1,2}$ be a circular ordering of $X \cup A_1 \cup A_2$ that extends $\sigma_1$ and $\sigma_2$ and that satisfies $\varphi$.
    Consider a direct angle $\alpha$ of $\sigma_1$ that is in the image of $\varphi$.
    Because $\sigma_{1,2}$ satisfies $\varphi$, $\alpha$ contains precisely the elements of $\varphi^{-1}(\alpha)$ in $\sigma_{1,2}$.
    It is therefore necessary that the elements of $\varphi^{-1}(\alpha)$ are consecutive in $\sigma_2$.   
    Because $|X| \geq 1$, $\sigma_2$ contains at least one element that does not belong to $\varphi^{-1}(\alpha)$ and thus $\sigma_2$ induces a unique linear ordering $\sigma_\alpha$ of the consecutive elements $\varphi^{-1}(\alpha)$.
    Since $\sigma_{1,2}$ extends $\sigma_2$, the elements $\varphi^{-1}(\alpha)$ appear in the linear order $\sigma_\alpha$ in the angle $\alpha$ in $\sigma_{1,2}$.
    For each direct angle that is in the image of $\varphi$, the order of the elements inside the angle is thus uniquely defined by $\sigma_2$.
    Hence $\sigma_{1,2}$ is unique.
\end{prooflater}

\subparagraph{\SEFE.}

Let $\mc I = \Instance$ be an instance of \SEFE.
Recall that $V(G_1) = \dots = V(G_k) \eqqcolon V$ and that all pairs of input graphs share the same subgraph.
We therefore let $G_\cap \coloneqq (V, \bigcap_{i \in [k]}E(G_i))$ denote the \emph{shared graph}.
Moreover, we let $\GU \coloneqq (V, \bigcup_{i \in [k]}E(G_i))$ denote the \emph{union graph}.
 For a subgraph $H_\cup$ of $\GU$, we let $H_i \coloneqq (V(H_\cup), E(G_i) \cap E(H_\cup))$ for $i \in [k]$ and $H_\cap \coloneqq (V(H_\cup), E(G_\cap) \cap E(H_\cup))$ denote its restriction to individual input graphs and the shared graph, respectively.
We let $|\mc I| = |V(\GU)|$ denote the \emph{size} of $\mc I$.

Jünger and Schulz~\cite{DBLP:journals/jgaa/JungerS09} showed that $\mc I$ admits a simultaneous drawing if and only if there exists a \emph{simultaneous embedding} $\embU = (\emb_1, \dots, \emb_k)$ of $\mc I$, that is, a sequence of planar embeddings of $G_1, \dots, G_k$ that induce the same embeddings on the shared graph $G_\cap$\footnote{Jünger and Schulz~\cite[Theorem 4]{DBLP:journals/jgaa/JungerS09} show this equivalence for $k = 2$ input graphs, but it is well-known that the same construction can also be successively applied to additional input graphs.}.

An instance $\mc I' = (G_1', \dots, G_k')$ of \SEFE is a \emph{subinstance} of $\mc I$ if $G_i'$ is a subgraph of $G_i$ for all $i \in [k]$.
In this case, we also refer to $\mc I$ as a \emph{superinstance} of $\mc I'$.
An isomorphism $f$ between \GU and $\GU'$ is a \emph{simultaneous isomorphism} if it is an isomorphism between $G_i$ and $G'_i$ for all $i \in [k]$.
If such an isomorphism exists, we say that the two instances are \emph{simultaneously isomorphic}.

We say that $\mc I$ is \emph{exclusive-connected} if $G_i$ is connected for all $i \in [k]$.
Bläsius and Rutter~\cite{BlasiusR15} showed that, for every instance of \SEFE with $k = 2$, an equivalent set of exclusive-connected instances can be computed in linear time.
The same construction works for $k \geq 3$ and with the following lemma, we show that the vertex integrity plus the number of input graphs of the resulting instance remains bounded.

\begin{restatable}\restateref{lm:exclconnected}{lemma}{lemmaExclConnected}
  \label{lm:exclconnected}
  Let $\mc I = \Instance$ be an instance of \SEFE.
  In polynomial time, we can compute a set of equivalent exclusive-connected instances of \SEFE, where each instance has $k$ input graphs and a union graph with vertex integrity in $O(k \cdot \vi\GU)$.
  
\end{restatable}

\begin{prooflater}{plemmaExclConnected}
        Note that multiple connected components of \GU can be decomposed into independent instances without increasing the vertex integrity.
  In the following, we thus assume that \GU is connected.
  Consider an edge $e = v_1v_2$ whose endpoints belong to different connected components of $G_i$ for some $i \in [k]$.
  Consider the instance $\hat {\mc I} = (\hat G_1, \dots, \hat G_k)$ obtained by introducing a new vertex $v$, adding the edge $v_1v$ exclusively to $G_i$, and the edge $vv_2$ as a shared edge.
  Note that this decreases the number of connected components of $G_i$ by one, while the number of connected components in all other graphs is unaffected.
  
  Clearly, any solution of the new instance also induces a solution of $\mc I$.
  For the other direction, the proof of Bläsius and Rutter~\cite[Lemma 1]{BlasiusR15} for $k = 2$ input graphs can also be applied verbatim for $k \geq 3$; we only briefly sketch the idea here and refer to their proof for details.
  Roughly speaking, for a given simultaneous embedding of $\mc I$, the existence of edge $e = v_1v_2$ in $\GU$ implies that $v_1$ and $v_2$ lie on the boundary of a common face of the shared graph.
  The fact that $v_1$ and $v_2$ belong to different connected components $H_i^1$ and $H_i^2$ of $G_i$ thus allows us to change the embeddings of $H_i^1$ and $H_i^2$ such that $v_1$ and $v_2$ subsequently share a face, without changing the embedding of the shared graph (see also Figure~2 in~\cite{BlasiusR15}).
  Hence $v$ and its two incident edges $v_1v$ and $vv_2$ can be embedded into this face.
  For all other graphs, $v$ has degree 1 and its incident edge can therefore be placed at the same position as in $G_i$, which yields a simultaneous embedding.
  Hence the two instances are equivalent.

  For $i \in [k]$, let $K_i \subseteq E(\GU) \setminus E(G_i)$  be the edges of a spanning tree of the graph obtained from $\GU$ by contracting each component of $G_i$ into a single vertex.
  Let $\mathcal I' = (G_1', \dots, G_k')$ be the instance obtained from $\mc I$ by inserting the gadget described above once for each $i \in [k]$ and each $e \in K_i$.
  Note that, after each gadget-insertion, the remaining tree edges of $K_i$ still join two different connected components in the $i$-th input graph.
  Moreover, as the gadget corresponds to a degree-1 vertex in all other input graphs, the connected components of all other graphs are unaffected.
  Hence $\mc I'$ is exclusive-connected and equivalent to $\mc I$.
  Moreover, $\GU'$ can be obtained from $\GU$ by inserting one or more parallel subdivided copies of each edge in $K \coloneqq \bigcup_{i \in [k]} K_i$.
  
  It remains to show that the vertex integrity of $\mc I'$ remains bounded.
  Let $X \subseteq V(\GU)$ be a witness of $\vi\GU$.
  Let $X' \subseteq V(\GU')$ be the set of vertices that consists of $X$ and all new degree-2 vertices whose neighbors both lie in $X$.
  We show that $X'$ witnesses $\vi{\GU'} \leq \vi\GU + (\vi\GU - 1) \cdot k$.

  Let $C'$ be a connected component of $\GU' - X'$.
  Note that $C \coloneqq C' \cap V(\GU)$ is a connected component of $\GU - X$, as the two neighbors of every new vertex are adjacent in $\GU$ and every new vertex outside $X'$ has a neighbor in $V(\GU) \setminus X$.
  Moreover, $\GU'[X' \cup C']$ can be obtained from $\GU[X \cup C]$ by inserting one or more gadgets for each edge in $K \cap E(\GU[X \cup C])$ and each such insertion decreases the number of connected components in $G_i[X \cup C]$.
  Since $|X \cup C| \leq \vi{\GU}$, this can happen at most $\vi{\GU} - 1$ times for each input graph.
  Hence $|X' \cup C'| \leq \vi{\GU} + k \cdot (\vi{\GU} - 1) \in O(k \cdot \vi{\GU})$.
  If $X' = V(\GU')$, the same argument with $C' = C = \emptyset$ yields the same upper bound for $|X'|$ alone.
  Hence $\vi{\GU'} \in O(k \cdot \vi\GU)$.
\end{prooflater}

We will therefore mostly work with exclusive-connected instances of \SEFE.
For an instance $\mc I$ of \SEFE, a \emph{simultaneous rotation system} of $\mc I$ is a sequence $\embU = (\emb_1, \dots, \emb_k)$ of rotation systems of $G_1, \dots, G_k$ (which are not necessarily planar or consistent on shared edges).
For exclusive-connected instances, every simultaneous embedding is uniquely defined by a corresponding simultaneous rotation system.

\section{A Fixed-Parameter Algorithm for Vertex Integrity and $\boldsymbol{k}$}
\label{section:vertex-integrity}

In this section, we show that \SEFE is FPT parameterized by the vertex integrity of the union graph $\GU$ plus the number of input graphs $k$.
Let $\mc I = \Instance$ be an instance of \SEFE and let $X$ be a set of vertices of \GU witnessing $\vi\GU$.
As stated in the introduction, \Cref{fig:noextension} shows that the standard argumentation for vertex integrity parameterizations does not work, as additional copies of (simultaneously) isomorphic components of $\GU - X$ cannot necessarily be added to a solution, no matter how many copies one keeps.

To tackle this issue, we employ the \emph{Ramsey pruning} technique that was recently introduced by Depian et al.~\cite{DFGS.LLR.2025}.
Roughly speaking, we first modify the input and $X$ in such a way that we can partition \GU into $X$ and a set of bounded-size \emph{blobs} that are pairwise simultaneously isomorphic. 
We arbitrarily order all blobs and if $\GU$ is larger than some function $f(\vi\GU, k)$, we can then use Ramsey-type arguments to show that any solution of the instance contains a (sufficiently large) subset of blobs that are \emph{$2$-consistent}, i.e., all increasing ordered pairs of blobs induce the same embedding in every input graph in this solution (similar to \Cref{fig:noextension}(c)).
 We then show that, if we restrict the solution to these $2$-consistent blobs, we can add an arbitrary number of additional blobs to the solution while preserving $2$-consistency.
 It thus suffices to keep $f(\vi\GU,k)$ blobs, as additional blobs can then always be reinserted.
Since the blobs have bounded size, this yields an equivalent bounded-size instance.

\begin{figure}
\centering
\includegraphics[scale=0.9]{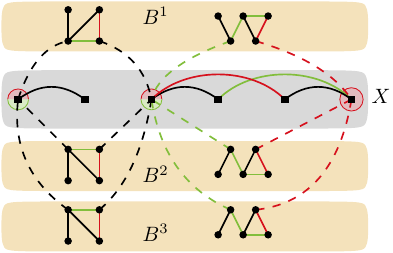}
\caption{
    An instance $\mc I$ of \SEFE on two graphs (green and red; shared edges are black) with blob-partition $\mc K = (X \dcup B^1 \dcup B^2 \dcup B^3, (\iso^{j_1\rightarrow j_2})_{j_1,j_2 \in [3]})$ with an extended modulator.
    Attachment edges are dashed and attachment vertices of the two graphs are highlighted with colored circles.
}
\label{fig:blob-partition}
\end{figure}

In the following, we first define these partitions; see also \Cref{fig:blob-partition}.
For an isomorphism~$f$ and a subset $X$ of its domain, we say that $f$ is \emph{$X$-preserving} if $f(x) = x$ for all $x \in X$.

\begin{definition}[Blob-partition]
  \label{def:blob-partition}
  Let $\mc I = \Instance$ be an instance of \SEFE.
  A \emph{blob-partition} of $\mc I$ is a pair $\mc K = (X \dcup B^1 \dcup \dots \dcup B^l, (\iso^{j_1\rightarrow j_2})_{j_1,j_2 \in [l]})$, where
  \begin{enumerate}
         \item $X \dcup B^1 \dcup \dots \dcup B^l = V(\GU)$,
    \item $N_{\GU}(B^j) \subseteq X$ for all $j \in [l]$,
    \item for all $j_1, j_2 \in [l]$, $\iso^{j_1\rightarrow j_2}$ is an $X$-preserving simultaneous isomorphism between $\GU[X \cup B^{j_1}]$ and $\GU[X \cup B^{j_2}]$, and
    \item for all $j_1, j_2, j_3 \in [l]$, $\iso^{j_1 \rightarrow j_3} = \iso^{j_2 \rightarrow j_3} \circ \iso^{j_1 \rightarrow j_2}$.
  \end{enumerate}
\end{definition}

 We refer to $X$ as the \emph{modulator} of $\mc K$ and to $B^1$, \dots, $B^l$ as \emph{blobs}.
Note that we allow the blobs to be empty sets, in which case $X = V(\GU)$.
We let $l$ denote the \emph{size} of $\mc K$.
We classify each edge $e$ of \GU as a \emph{modulator edge}, an \emph{attachment edge}, or a \emph{blob edge}, based on whether $e$ has two, one, or zero endpoints in $X$, respectively.
For a set $S \subseteq [l]$, we let $\mc I^{S} = (G_1^S, \dots, G_k^S)$ denote the subinstance of $\mc I$ induced by vertices $X \cup \bigcup_{j \in S}B^j$.
We let $\mc K^S = (X \dcup \mathop{\dot\bigcup}_{j \in S} B^j,(\iso^{j_1 \rightarrow j_2})_{j_1,j_2 \in S})$ denote the corresponding blob-partition of $\mc I^S$.

Since the isomorphisms $(\iso^{j_1\rightarrow j_2})_{j_1,j_2 \in [l]}$ are $X$-preserving, we have $N_{\GU}(B^{j_1}) = N_{\GU}(B^{j_2}) \subseteq X$ for all $j_1, j_2 \in [l]$.
We let $\att(\mc K) = N_{\GU}(B^1) = \dots = N_{\GU}(B^l)$ denote the \emph{attachment vertices} of $\mc K$.
For $i \in [k]$, we let $\att_i(\mc K) = N_{G_i}(B^1) = \dots = N_{G_i}(B^l) \subseteq \att(\mc K)$ denote the attachment vertices that are incident to an attachment edge of $G_i$.
  We say that the modulator $X$ is \emph{extended} if 
 for each $i \in [k]$ and each $x \in \att_i(\mc K)$, $x$ is incident to at least one modulator edge in $G_i$.
This property will be crucial, as it guarantees that rotation systems of different blobs always have a common reference point.
We can easily show that exclusive-connectedness and extended modulators also translate to subsets of blobs.

\begin{restatable}\restateref{lm:connectedSubinstances}{lemma}{lemmaConnectedSubinstances}
    \label{lm:connectedSubinstances}
    Let $\mc I$ be an exclusive-connected instance of \SEFE on $k$ input graphs and let $\mc K$ be a blob-partition of $\mc I$ with an extended modulator $X$ and size $l$.
    Then for every non-empty $S \subseteq [l]$, $X$ is also extended in $\mc K^S$ and $\mc I^S$ is exclusive-connected.
\end{restatable}
\begin{prooflater}{plemmaConnectedSubinstances}
    Since the attachment vertices are the same in $\mathcal K$ and $\mathcal K^S$, $X$ remains extended in $\mc K^S$.

    Let $i \in [k]$ and let $u, v \in V(G^S_i)$.
    Since $G_i$ is connected, it contains a $uv$-path $P$.
    If $P$ visits $B^j$ for some $j \notin S$, observe that we can reroute the corresponding part of $P$ through an arbitrary other blob $B^{j'}$ for $j' \in S$ using $\iso^{j \rightarrow j'}$.
    Since $S \neq \emptyset$, we thus also find a $uv$-path in $G^S_i$.
    Hence $G^S_i$ is connected for each $i \in [k]$ and thus $\mc I^S$ is exclusive-connected.
\end{prooflater}

 Assume $\mathcal I$ is exclusive-connected and $X$ is extended.
Let $\embU = (\emb_1, \dots, \emb_k)$ be a simultaneous rotation system of $\mc I$.
For a set $S \subseteq [l]$, we let $\embU^S = (\emb^S_1, \dots, \emb^S_k)$ denote the restriction of $\embU$ to $\mc I^S$.
For $S \neq \emptyset$, $\mc I^S$ remains exclusive-connected by \Cref{lm:connectedSubinstances}. 
Hence the rotation systems $\emb^S_1, \dots, \emb^S_k$ define embeddings of $G^S_1, \dots, G^S_k$, as no relative positions are involved.
 
Let $m \in [l]$, let $r = (r_1, \dots, r_m)$, $ s = (s_1, \dots, s_m)$ be two strictly increasing sequences over~$[l]$.
Since the simultaneous isomorphisms $(\iso^{j_1\rightarrow j_2})_{j_1,j_2 \in [l]}$ are $X$-preserving, note that $\iso^{r \rightarrow  s} \coloneqq \iso^{r_1 \rightarrow s_1} \cup \dots \cup \iso^{r_m \rightarrow s_m}$ is well-defined and a simultaneous isomorphism between $\GU^{\{r_1, \dots, r_m\}}$ and $\GU^{\{s_1, \dots s_m\}}$.
For $i \in [k]$, we write $\emb_i^{\{s_1, \dots, s_m\}} = \iso^{r\rightarrow s}(\emb_i^{\{r_1, \dots, r_m\}})$ if 
$\emb_i^{\{s_1, \dots, s_m\}}(\iso^{r\rightarrow s}(x))= \iso^{r \rightarrow s}(\emb_i^{\{r_1, \dots, r_m\}}(x))$
holds for every $x \in V(\GU^{\{r_1, \dots, r_m\}})$. 

Let $S \subseteq [l]$, let $m \in [|S|]$, and let $r = (r_1, \dots, r_m)$ be the strictly increasing sequence induced by the $m$ smallest indices in $S$.
We say that $\embU^S$ is \emph{$m$-consistent} (with respect to $\mathcal K^S$) if, for every $i \in [k]$ and every strictly increasing sequence $s = (s_1, \dots, s_m)$ over $S$, we have $\emb_i^{\{s_1, \dots, s_m\}} = \iso^{r\rightarrow s}(\emb_i^{\{r_1, \dots, r_m\}})$.
In other words, any $m$ blobs of $\mathcal K^S$ look the same as the first $m$ blobs in $\embU^S$.
For $S = [l]$, we simply say $\embU$ is $m$-consistent with respect to $\mathcal K$.

Our goal is to prove the following lemma.
\begin{restatable}{lemma}{addBlob}
  \label{lm:addBlob}
  Let $\mc I = \Instance$ be an exclusive-connected instance of \SEFE.
  Let $\mc K$ be a blob-partition of $\mc I$ with an extended modulator $X$ and size $l \geq |X| + 4$.
  
  If there exists a simultaneous embedding $\embPU$ of $\mc I^{[l-1]}$ that is $2$-consistent with respect to $\mc K^{[l-1]}$, there  exists a simultaneous embedding $\embU$ of $\mc I$ that is $2$-consistent with respect to $\mc K$.
\end{restatable}

\begin{figure}[t]
\begin{subfigure}{\textwidth}
    \centering
    \input{figures/ramsey-overview}
    \phantomsubcaption\label{fig:strategyOverview}
	\phantomsubcaption\label{fig:computeReduced}
\end{subfigure}
\caption{
    \textbf{(a)} The overall strategy to test whether a given exclusive-connected instance $\mathcal I$ of \SEFE admits a simultaneous embedding.
    \textbf{(b)} Illustration of \Cref{lm:computeReduced}: Computing the blob-partition $\mc K = (X_m \dcup B^1 \dcup \dots \dcup B^{l_m}, (\iso^{j_1\rightarrow j_2})_{j_1,j_2 \in [l_m]})$ of a reduced instance $\Ir$, given a witness $S$ of $\vi\GU$.
    The (clickable) references in the figure link to the formal lemma statement.
}
\label{fig:overview}
\end{figure}

As the proof of this lemma is rather involved, we defer it to \Cref{sec:blobProof}.
We now use this lemma to obtain the desired bounded-size equivalent instance; see also \Cref{fig:strategyOverview}.
For a given exclusive-connected instance $\mc I$ of \SEFE, we will characterize equivalent \emph{reduced} subinstances on whose solutions we can apply \Cref{lm:addBlob} to obtain a solution of the original instance.
Such a reduced instance must admit a special blob-partition that (1) has a sufficient number of blobs for us to use Ramsey-theory to argue that any simultaneous embedding has $|X| + 3$ blobs that are $2$-consistent and (2) ``blowing up'' the instance by adding a finite number of additional blobs again yields (a superinstance of) the original instance.

We now first define the ``blow-up'' operation and the necessary functions to lower-bound the number of blobs, and we subsequently define reduced instances.

Let $\mc I = \Instance$ be an instance of \SEFE with a fixed blob-partition $\mc K = (X \dcup B^1 \dcup \dots \dcup B^{l}, (\iso^{j_1\rightarrow j_2})_{j_1,j_2 \in [l]})$.
For $t \in \mathbb N$, we let $\mc I^{+t}$ be the instance with blob-partition $\mc K^{+t} = (X \dcup B^1 \dcup \dots \dcup B^{l + t}, (\iso^{j_1\rightarrow j_2})_{j_1,j_2 \in [l + t]})$  obtained from $\mc I$ and $\mc K$ by adding $t$ additional blobs to the instance.
Let the multicolor Ramsey-number $R(a,b)$ denote the smallest integer~$n$ such that any edge-coloring of $K_n$ with $b$ colors contains a monochromatic clique of size at least $a$.
Moreover, let $g(n, k)$ denote the maximum number of distinct simultaneous rotation systems an $n$-vertex \SEFE-instance on $k$ input graphs can have.
\begin{definition}[Reduced instance]
    \label{def:reduced}
    Let $\mc I$ be an exclusive-connected instance of \SEFE on $k$ graphs.
    An exclusive-connected subinstance $\Ir$ of $\mc I$ is \emph{reduced} with respect to $\mc I$ if \Ir admits a blob-partition $\mc K = (X \dcup B^1 \dcup \dots \dcup B^{l}, (\iso^{j_1\rightarrow j_2})_{j_1,j_2 \in [l]})$ such that
    \begin{enumerate}
        \item\label{it:extended} the modulator $X$ of $\mc K$ is extended,
        \item\label{it:size} $l \geq R(|X|+3,g(|X| + 2|B^1|,k))$, and
        \item\label{it:blowup} for some $t \in \mathbb N$, $\Ir^{+t}$ is simultaneously isomorphic to a superinstance of $\mc I$.
             \end{enumerate}
\end{definition}

Roughly speaking, the choice of $g(|X| + 2|B^1|, k)$ for the number of colors in Property~\ref{it:size} upper-bounds the number of distinct simultaneous rotation systems for a pair of blobs.
The chosen value of $l$ thus guarantees that, in every solution of $\Ir$, we find $|X| + 3$ blobs that induce a $2$-consistent solution.
Since $\Ir$ is exclusive-connected and the modulator of $\mathcal K$ is extended (Property~\ref{it:extended}), we can use \Cref{lm:addBlob} to repeatedly insert additional blobs into this solution.
By Property~\ref{it:blowup}, we can reach a solution of an instance that is simultaneously isomorphic to a superinstance of $\mc I$.
This also induces a solution of $\mc I$.
Using this strategy, we can show that every reduced instance is equivalent to its original instance; see~\Cref{fig:strategyOverview}.

\begin{restatable}\restateref{lm:reducedEquivalent}{lemma}{lemmaReducedEquivalent}
    \label{lm:reducedEquivalent}
    Let $\mc I$ be an exclusive-connected instance of \SEFE.
    Let \Ir be an instance that is reduced with respect to $\mc I$.
    Then $\mc I$ is a yes-instance if and only if \Ir is.
\end{restatable}
\begin{prooflater}{plemmaReducedEquivalent}
    Since \Ir is a subinstance of $\mc I$, any solution of $\mc I$ also induces a solution of \Ir.

    Conversely, let $\embU$ be a simultaneous embedding of \Ir.
    Since \Ir is reduced with respect to $\mc I$, it is an exclusive-connected subinstance of $\mc I$ and there exists a blob-partition $\mc K = (X \dcup B^1 \dcup \dots \dcup B^{l}, (\iso^{j_1\rightarrow j_2})_{j_1,j_2 \in [l]})$ of \Ir that satisfies Properties~\ref{it:extended}--\ref{it:blowup} of \Cref{def:reduced}.

    \begin{claim}
        There exists a subset $S \subseteq [l]$ of size $|X| + 3$ such that $\embU^S$ is $2$-consistent with respect to $\mc K^S$.
    \end{claim}
    \begin{claimproof}
        Recall that all blobs of $\mc K$ have the same size and thus the number of possible simultaneous rotation systems of  $\GU[X \cup B^{1} \cup B^{2}]$ is at most $g(|X \cup B^{1} \cup B^{2}|,k) = g(|X| + 2|B^1|,k)$.
        Consider the complete graph $K_l$ on $l$ vertices that correspond to the blobs of $\mc K$.
        For all $j_1 < j_2 \in [l]$, color the edge between $B^{j_1}$ and $B^{j_2}$ with $\iso^{(j_1,j_2) \rightarrow (1,2)}(\embU^{\{j_1,j_2\}})$.
        As argued above, this yields at most $g(|X| + 2|B^1|,k)$ colors.
        Since $l \geq R(|X|+3,g(|X| + 2|B^1|,k))$, we find a monochromatic clique $S$ of size $|X|+3$ in $K_l$.
        For all $i \in [k]$ and all $a_1 < a_2 \in S$, $b_1 < b_2 \in S$, we thus have $\iso^{(a_1,a_2) \rightarrow (1,2)}(\embU^{\{a_1,a_2\}}) = \iso^{(b_1,b_2) \rightarrow (1,2)}(\embU^{\{b_1,b_2\}})$ and thus $\embU^{\{a_1,a_2\}} = \iso^{(b_1,b_2) \rightarrow (a_1,a_2)}(\embU^{\{b_1,b_2\}})$.
        Hence $\embU^S$ is $2$-consistent with respect to $\mc K^S$.
             \end{claimproof}

    By order-preservingly renaming the blobs of $\mc K$, we can assume without loss of generality that $\embU^{[|X|+3]}$ is $2$-consistent with respect to $\mc K^{[|X|+3]}$.
    By Property~\ref{it:blowup}, there exists a $t \in \mathbb N$ such that the instance $\Is \coloneqq \Ir^{+t}$ is simultaneously isomorphic to a superinstance of $\mc I$.
    After renaming the vertices of $\Is$ accordingly, we can assume that $\Is$ is a superinstance of $\mc I$.
    Let $\mc K' \coloneqq \mc K^{+t}$ be the corresponding blob-partition of $\Is$.
    Recall that $\Ir$ is exclusive-connected and the modulator $X$ of $\mc K$ is extended.
    Observe that adding additional blobs retains both properties and thus $\Is$ is exclusive-connected and $X$ is also extended in $\mc K'$.

    \begin{claim}
        Let $|X| +3 \leq m \leq l + t$.
        Then $\Is^{[m]}$ admits a simultaneous embedding that is $2$-consistent with respect to $\mc K'^{[m]}$.
             \end{claim}
    \begin{claimproof}
        For $m = |X| + 3$, we have $\mc K'^{[m]} = \mc K^{[|X| + 3]}$ and thus the statement holds as argued above.
        Let $m > |X| + 3$.
        By induction, the instance $\Is^{[m-1]}$ admits a simultaneous embedding $\embPU$ that is $2$-consistent with respect to  $\mc K'^{[m-1]}$.
        By \Cref{lm:connectedSubinstances}, $\Is^{[m]}$ is exclusive-connected and $X$ is extended in $\mc K'^{[m]}$.
        We can thus apply \Cref{lm:addBlob} and obtain the desired simultaneous embedding of $\Is^{[m]}$ that is $2$-consistent with respect to $\mc K'^{[m]}$.
    \end{claimproof}
    
    Hence the instance $\Is = \Is^{[l + t]}$ admits a simultaneous embedding and is therefore a yes-instance.
    Since $\mc I$ is a subinstance of \Is, the same holds for $\mc I$, which concludes the proof.
\end{prooflater}

As the final step, we now show that we can compute a bounded-size reduced instance in fixed-parameter time; see \Cref{fig:computeReduced} for an illustration.

\begin{restatable}\restateref{lm:computeReduced}{lemma}{lemmaComputeReduced}
    \label{lm:computeReduced}
    There is an algorithm that takes as input an exclusive-connected instance $\mc I = \Instance$ of \SEFE and computes an instance of size at most $f(\vi\GU, k)$ that is reduced with respect to $\mc I$ in time $f(\vi\GU, k) \cdot {|V(\GU)|}^{O(1)}$ with $f(\vi\GU, k)  \in 2\upup O(\vi\GU^2 \cdot (k+2)^{\vi\GU^2+1})$.
\end{restatable}
\begin{proofsketch}
    We start with a witness $S \subseteq V(\GU)$ of $\vi\GU$, i.e., we have $|S \cup C| \leq \vi\GU$ for all components $C$ of $\GU - S$.
    We group these components into a bounded number of equivalence classes based on $S$-preserving simultaneous isomorphism.
    Using $S$ and these equivalence classes, we construct the desired reduced instance with a suitable blob-partition.

    Roughly speaking, the goal is to compute an extended modulator $X_m$ such that all remaining equivalence classes that are not contained in $X_m$ are sufficiently large for us to split them into sufficiently many blobs (Property~\ref{it:size}); see \Cref{fig:computeReduced}.
    First, we extend $S$ to a set $S'$ by adding one representative of each equivalence class; this will ensure that the resulting modulator is extended.
    We then start with $X_0 = S'$ and for $i > 0$ we iteratively compute a growing set $X_i$ as follows.
    If there remains an equivalence class of size less than $R(|X_{i-1}|+3,g(3|X_{i-1}|,k))$, we obtain the set $X_{i}$ by adding this equivalence class to $X_{i-1}$.
    If no such equivalence class remains, the procedure stops with the desired set $X_m$.
    
    Since $X_i$ grows in every step, so does the required bound, which is why we only move one equivalence class at a time into the modulator.
    Because the number of equivalence classes is bounded, this process stops after a bounded number of steps and the resulting set $X_m$ has bounded size.
    We then compute a subinstance \Ir of $\mc I$ together with a corresponding blob-partition $\mc K = (X_m \dcup B^1 \dcup \dots \dcup B^l, (\iso^{j_1\rightarrow j_2})_{j_1,j_2 \in [l]})$ by using $X_m$ as the modulator and constructing $l = R(|X_m|+3,g(3|X_m|,k))$ blobs that consist of precisely one copy of each remaining equivalence class; see \Cref{fig:computeReduced}.
    Note that the size of each blob is bounded.

    We then show that \Ir is reduced.
    The instance \Ir inherits exclusive-connectedness from $\mc I$ and the modulator $X_m$ is extended by construction.
         Because~$X_m$ contains one copy of each equivalence class by construction, we have $|X_m| \geq |B^1|$ and thus also $l  \geq R(|X_m|+3,g(|X_m| + 2|B^1|,k))$.
    The blob-partition therefore has sufficient size.
    Moreover, observe that the difference between \Ir and $\mc I$ lies in copies of equivalence classes whose sizes are at least $l$.
    By adding a sufficient number of copies of blobs to \Ir, we thus obtain an instance that is (simultaneously isomorphic to) a superinstance of $\mc I$.
    Thus \Ir is reduced.

    Since the size of $X_m$, $l$, and each blob is bounded in $\vi\GU$ and $k$, the same holds for $|\Ir|$.
    Note, however, that the iterative computation of the modulator yields a tower of exponentials that also has exponential height.
    Moreover, we can compute \Ir in fixed-parameter time.
\end{proofsketch}
\begin{prooflater}{plemmaComputeReduced}
    Let $S \subseteq V(\GU)$ be a witness of $\vi\GU \eqqcolon \psi$ and let $\mc C$ be the set of connected components of $\GU - S$.
    Let $\sim$ denote the equivalence relation on $\mc C$ where two connected components $C_1, C_2 \in \mc C$ are equivalent if there exists an $S$-preserving simultaneous isomorphism between $\GU[S \cup C_1]$ and $\GU[S \cup C_2]$.
    Recall that we have $|S \cup C| \leq \psi$ for each $C \in \mc C$.

    Consider instances of \SEFE that contain $S$ and a total of at most $\psi$ vertices.
    Observe that there are at most $\psi \cdot (k+2)^{\psi^2}$ distinct instances of \SEFE (up to $S$-preserving simultaneous isomorphism), as every pair of vertices is either a non-edge, a shared edge, or an edge in exactly one of the $k$ input graphs.
    Thus, we can upper-bound the number of equivalence classes of $\sim$ by $h(\psi, k) = \psi \cdot (k+2)^{\psi^2}$.
     
    We now employ the following strategy; see also \Cref{fig:computeReduced}.
    To obtain a suitable blob-partition of a reduced instance, our goal is to form blobs which consist of precisely one copy of each equivalence class.
    Since the blob-partition must have an extended modulator, we will first add one copy of each equivalence class to $S$.
    As we also need a sufficient number of blobs, we then remove equivalence classes of insufficient size by adding them to the modulator.
    As this increases the size of the modulator and thus also of the required number of blobs, we have to repeat this step until all remaining equivalence classes are sufficiently large (or until all of them are part of the modulator, in which case the entire graph has bounded size).

    We first fix a function $u(x,k)$ that upper-bounds $R(x + 3, g(3x,k))$.
    It is known that $R(a,b) \leq b^{ab}$~\cite{BBCGHMST26} and observe that $g(n, k)$ is trivially upper-bounded by $(n!)^{nk} \leq n^{n^2k}$.
    We thus choose $u(x,k) \coloneqq (3x)^{(3x)^2\cdot k\cdot (x+3) \cdot (3x)^{(3x)^2 \cdot k}} \geq R(x + 3, g(3x,k))$.

    Let $K \subseteq \mc C$ contain precisely one representative of each equivalence class.
    Let $S' = S \cup \bigcup_{C \in K} C$ and $\mc C' = \mc C \setminus K$.
    We now inductively define sets $X_i$ and $\mc C_i$ for all $i \in \mathbb N$ as follows.
    For $i = 0$, we set $X_0 \coloneqq S'$ and $\mc C_0 = \mc C'$.
         For $i > 0$, let $l_i \coloneqq u(|X_{i-1}|,k)$.
    If there exists an equivalence class $[C] \in \mc C_{i-1} / {\sim}$ with $|[C]| < l_i$, we set $X_i \coloneqq X_{i-1} \cup \bigcup_{C' \in [C]} C'$ and $\mc C_i = \mc C_{i-1} \setminus [C]$.
    If no such equivalence class exists, we set $X_i \coloneqq X_{i-1}$ and $\mc C_i \coloneqq \mc C_{i-1}$.
    
    Once the latter case occurs, observe that all subsequent steps also leave the two sets unchanged.
    Since we have at most $h(\psi, k)$ equivalence classes, this process converges after at most $h(\psi, k)$ steps and thus for $m = h(\psi, k) + 1$, we have $X_m = X_{m-1}$ and every equivalence class of $\mc C_m / {\sim}$ has size 
         at least $l_m = u(|X_m|, k)$.

    Let $\mc K = (X_m \dcup B^1 \dcup \dots \dcup B^{l_m}, (\iso^{j_1\rightarrow j_2})_{j_1,j_2 \in [l_m]})$ be the blob-partition of a subinstance $\Ir \coloneqq (G_1', \dots, G_k')$  of $\mc I$ consisting of $l_m$ blobs, each of which contains precisely one copy of each equivalence class of   $\mc C_m / {\sim}$.
    For the isomorphisms $(\iso^{j_1\rightarrow j_2})_{j_1,j_2 \in [l_m]}$ of $\mathcal K$, note that the $S$-preserving simultaneous isomorphisms for equivalent components of $\GU - S$ can be extended to $X_m$-preserving simultaneous isomorphisms for equivalent components of $\GU - X_m$.
    These can in turn be combined to $X_m$-preserving simultaneous isomorphisms for the blobs.
    To ensure their consistency, first fix an isomorphism $\iso^{1 \rightarrow j}$ for all $j \in [l_m]$ (where $\iso^{1 \rightarrow 1}$ is the identity) and set $\iso^{j_1 \rightarrow j_2} \coloneqq \iso^{1 \rightarrow j_2} \circ ({\iso^{1 \rightarrow j_1}})^{-1}$.
    This way, $\mathcal K$ is indeed a blob-partition.
    Our next goal is to show that $\Ir$ is reduced with respect to $\mc I$ by verifying the requirements of \Cref{def:reduced}.

    We first show that $\Ir$ remains exclusive-connected using a similar argument as the proof of \Cref{lm:connectedSubinstances}.
    Let $i \in [k]$ and $u, v \in V(G'_i)$.
    Because $G_i$ is connected, there exists a $uv$-path $P$ in $G_i$.
    Assume that $P$ contains vertices of a connected component $C \in \mathcal C_m \subseteq \mathcal C$ that is not part of $\Ir$.
    Consider a maximal subpath $P'$ of $P$ whose internal vertices all belong to $C$.
    Since the endpoints of $P$ are not contained in $C$, $N_{G_i}(C) \subseteq S$ implies that the endpoints of $P'$ both lie in $S$.
    There exists an equivalent connected component $C' \sim C$ of $\mathcal C / {\sim}$ with $C' \subseteq S'$, because we added one representative of each equivalence class to $S'$.
    Since $C'$ and $C$ are equivalent, there is an $S$-preserving isomorphism between $G_i[S \cup C']$ and $G_i[S \cup C]$.
    Because $C', S \subseteq S' \subseteq X_m$, we have $S \cup C' \subseteq X_m \subseteq V(G'_i)$.
    We can thus reroute $P'$ through $G_i[X_m] = G'_i[X_m]$.
         Hence $G'_i$ is connected.

    We now argue that the modulator $X_m$ of $\mathcal K$ is extended (Property~\ref{it:extended}) in a similar fashion.
    To this end, let $i \in [k]$ and consider a vertex $v \in \att_i(\mc K) \subseteq X_m$.
    By the definition of $\att_i(\mc K)$, there exists a neighbor $u \notin X_m$ of $v$ in $G_i'$.
    Since $S \subseteq X_m$, we also have $u \notin S$ and thus there exists a connected component $C \in \mc{C}$ of $\GU - S$ with $u \in C$.
    Since $u \in C$ and $v \in X_m$ but $u \notin X_m$, we have $v \notin C$.
    Because $N_{G_i}(C) \subseteq S$, the existence of edge $uv$ thus implies that $v \in S$.
    As above, there exists an equivalent connected component $C' \sim C$ with $C' \subseteq X_m$.
    The corresponding $S$-preserving isomorphism between $G_i[S \cup C']$ and $G_i[S \cup C]$ yields a neighbor $u'$ of $v$ in $G_i[S \cup C']$.
    Since $S \cup C' \subseteq X_m$, $v$ thus has degree at least one in $G_i'[X_m] = G_i[X_m]$.
    Hence $X_m$ is extended.

    Towards showing $l_m \geq R(|X_m| + 3, g(|X_m| + 2|B^1|, k))$ (Property~\ref{it:size}), we now argue $|X_m| \geq |B^1|$.
    Observe that $B^1$ contains at most one copy of every equivalence class of $\mc{C}/{\sim}$ (since $\mc{C}_m/{\sim} \subseteq \mc{C}/{\sim}$).
    Moreover, recall that $K$ contains one copy of every equivalence class of $\mc{C}/{\sim}$ and ${K \subseteq S'}$.
    Finally, recall that $S' \subseteq X_m$, which implies $|X_m| \geq |B^1|$.
    Thus $l_m = u(|X_m|, k) \geq R(|X_m| + 3, g(3|X_m|, k)) \geq R(|X_m| + 3, g(|X_m| + 2|B^1|, k))$.

    Finally, consider Property~\ref{it:blowup}.
    Let $w = \max(\{l_m\} \cup \{|[C]| \mid [C] \in \mc{C}_m/{\sim}\})$ and observe $w \geq l_m$.
    Let $t = w - l_m$ and consider the instance $\Is = \Ir^{+t}$.
    Observe that $\Is$ contains $X_m$ and $w$ copies of each equivalence class of $\mc{C}_m/{\sim}$.
    Hence $\Is$ is simultaneously isomorphic to a superinstance of $\mc I$.
    Therefore, $\Ir$ is reduced with respect to $\mc I$.

    \medskip

    It remains to establish an upper bound on the size of $\Ir$ and the running time of the above-presented algorithm.
    
    \begin{claim}
        We have $|\Ir| \in 2\upup O(\psi^2 \cdot (k+2)^{\psi^2+1})$.
             \end{claim}
         \begin{claimproof}
                 We first inductively prove an upper bound for $|X_i|$ for all $i \in \mathbb N_0$.
        To this end, we first fix corresponding constants.
        Let $\expt{a}{b} = 2^{2^{\dots^{2^b}}}$ denote the $a$-fold concatenation of base-2 exponentiation.
        Recall that $X_0 = S'$ and $S'$ is obtained by adding a copy of every equivalence class to $S$.
        Thus, $|X_0| \leq \psi + \psi \cdot h(\psi, k) = \psi + \psi^2 \cdot (k+2)^{\psi^2}$.
        This can clearly be upper-bounded by $\expt{c_1}{k + \psi}$ for a small constant $c_1 \in \mathbb N$.
        Moreover, let $c_2 \in \mathbb N$ be a small constant such that $x + u(x, k) \cdot \psi \leq \expt{c_2 \cdot (k + \psi)}{x}$.
        Recall that $u$ is at most triple exponential in $x$ and $k$ and thus such a constant clearly exists.

        We now show inductively that $|X_i| \leq \expt{(i \cdot c_2 + c_1) \cdot (k + \psi)}{k + \psi}$.
        By the choice of $c_1$, this holds for $i = 0$.
        For $i > 0$, recall that we either have $X_i = X_{i-1}$ or $X_i$ is obtained from $X_{i-1}$ by adding all connected components of an equivalence class of size smaller than $l_i = u(|X_{i-1}|, k)$.
        Since each such connected component has size at most $\psi$, we have $|X_i| \leq |X_{i-1}| + u(|X_{i-1}|, k) \cdot \psi$.
        By the choice of $c_2$, we have 
        \begin{align*}
          |X_i| &\leq |X_{i-1}| + u(|X_{i-1}|, k) \cdot \psi\\
                &\leq\expt{c_2 \cdot (k+\psi)}{|X_{i-1}|} \\
                &\leq\expt{c_2 \cdot (k+\psi)}{\expt{((i-1) \cdot c_2 + c_1) \cdot (k+\psi)}{k+\psi}} \\
                &=   \expt{(c_2 + (i-1) \cdot c_2 + c_1) \cdot (k+\psi)}{k+\psi} \\
                &\leq\expt{(i \cdot c_2 + c_1) \cdot (k+\psi)}{k+\psi}
        \end{align*}

        Recall that $l_i = u(|X_{i-1}|, k)$.
        Hence also $l_i \leq \expt{c_2 \cdot (k + \psi)}{|X_{i-1}|}$.
        We thus analogously obtain $l_i \leq \expt{(i \cdot c_2 + c_1) \cdot (k + \psi)}{k + \psi}$.
        Hence $|X_m|, l_m \in \expt{O(m \cdot (k + \psi))}{k + \psi}$.

        We are now ready to bound $|\Ir|$.
        Observe that $|\Ir| = |X_m| + |B^1| + \dots + |B^{l_m}| = |X_m| + l_m \cdot |B^1|$.
        Recall from above that $|B^1| \leq |X_m|$ and thus $|\Ir| \leq |X_m| \cdot (l_m + 1)$.
        Using the computed bounds and the fact that the product only adds a constant number of levels in the tower of exponentials, we also obtain $|\Ir| \in  \expt{O(m \cdot (k + \psi))}{k + \psi}$.
        Since $m = h(\psi, k) + 1 = \psi \cdot (k+2)^{\psi^2} + 1$, we thus have $|\Ir| \in  \expt{O(\psi^2 \cdot (k+2)^{\psi^2+1})}{k + \psi}$.
        This is upper-bounded by $2\upup O(\psi^2 \cdot (k+2)^{\psi^2+1})$, which gives the desired bound.
    \end{claimproof}

    Note that the algorithm described above also computes $\Ir$ in time $2\upup O(\psi^2 \cdot (k+2)^{\psi^2+1}) \cdot |V(\GU)|^{O(1)}$.
    This includes computing the witness $S$ of $\vi\GU$, which can be done in $O(\psi^{\psi+1} \cdot |V(\GU)|)$ time~\cite{DrangeDH16}, and computing the equivalence classes of $\sim$.
               \end{prooflater}

Now let $\mc I = (G_1, \dots G_k)$ be an instance of \SEFE.
By \Cref{lm:exclconnected}, we can assume that $\mc I$ is exclusive-connected, as the corresponding construction keeps the vertex integrity of the union graph plus the number of input graphs bounded.
Using \Cref{lm:computeReduced}, we can then compute a bounded-size reduced instance of $\mc I$ in fixed-parameter time.
This instance can be solved by enumerating all simultaneous rotation systems and testing whether one of them is a simultaneous embedding.
Since the reduced instance is equivalent to the original instance by \Cref{lm:reducedEquivalent}, we thus obtain the following result.

\vertexintegrityfpt*

\subsection{Proof of \Cref{lm:addBlob}}
\label{sec:blobProof}

\addBlob*
 To construct the desired simultaneous embedding \embU of $\mc I$, we will employ the following strategy.
For each $i \in [k]$, note that $G_i$ is the union of $G_i^{[l-1]}$ and $G_i^{\{l\}}$.
We take the rotation system $\embP_i$ of $G_i^{[l-1]}$ and for $G_i^{\{l\}}$, we obtain a rotation system by ``copying'' the rotation system ${\embP_i}^{\{1\}}$ under the isomorphism $\iso^{1\rightarrow l}$.
Since $G_i$ is connected, to obtain an embedding $\emb_i$ of $G_i$, we then only have to specify how to interleave the rotations of these two embeddings at attachment vertices, as the two rotation systems already agree on modulator edges.
We have to do this in such a way that (1) $\embU$ is $2$-consistent with respect to $\mc K$, (2) each $\emb_i$ is a planar embedding, and (3) $\embU = (\emb_1, \dots, \emb_k)$ is a simultaneous embedding.
Proving that these properties can be maintained is non-trivial and we need several preliminary results.

We first characterize $2$-consistency in terms of angle-assignments (\Cref{lm:twoconsistent}) and use this to show that the $2$-consistent extension of $\embPU$ to $\mc I$ exists (\Cref{lm:angleAssignment,lm:twoConsEmb}).
We then use the fact that the modulator is extended to show that $2$-consistency also implies $m$-consistency for all $m \in [l]$ (\Cref{lm:m-consistent}).
Using this, together with the fact that obstructions to planarity (\Cref{lm:planarityCharacterization}) and simultaneous planarity appear in a small number of blobs, we then show that the constructed simultaneous rotation system is a simultaneous embedding of $\mc I$ (\Cref{lm:simEmb}).

Let $j \in [l]$ and let $v^j$ be a vertex of $\GU^{\{j\}}$.
For each $j' \in [l]$, we let $v^{j'} = \iso^{j \rightarrow j'}(v^j)$ denote the corresponding vertex in $\GU^{\{j'\}}$ (due to the fourth property of \Cref{def:blob-partition}, we have that $\iso^{j \rightarrow j}$ is the identity).
Observe that, for $j' \neq j$, we have $v^j = v^{j'}$ if and only if $v^j \in X$.
Analogously, we define $e^{j'}$ for every edge $e^j$ of $\GU^{\{j\}}$.
As above, observe that we have $e^{j'} = e^j$ for $j \neq j'$ if and only if $e^j$ is a modulator edge.

Let $ i \in [k]$ and let $x \in \att_i(\mc K)$ be an attachment vertex in $G_i$.
We let $X_i(x)$ denote the set of modulator edges incident to $x$ in $G_i$ and for $j \in [l]$, we let $A_i^j(x)$ denote the set of attachment edges incident to $x$ in $G_i^{\{j\}}$.
For $S \subseteq [l]$, observe that the edges incident to~$x$ in $G_i^S$ are  $X_i(x) \cup \bigcup_{j \in S}A_i^j(x)$.
Note that $|X_i(x)| \geq 1$, because $X$ is an extended modulator.

In the following, we first characterize $2$-consistency using angle-assignments.
To this end, we consider the direct angle-assignment induced by the first two blobs and map it to every other pair of blobs using the corresponding isomorphism; see \Cref{fig:transport}.
We can then characterize $2$-consistency in terms of these angle-assignments.

\begin{figure}
    \centering
    \begin{subfigure}{0.24\textwidth}
      \centering
      \includegraphics[page=1]{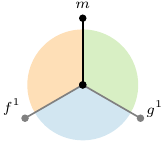}
      \caption{}
      \label{fig:transportA}
    \end{subfigure}
    \hfill
    \begin{subfigure}{0.24\textwidth}
      \centering
      \includegraphics[page=2]{figures/AngleTransport}
      \caption{}
      \label{fig:transportB}
    \end{subfigure}
    \hfill
    \begin{subfigure}{0.24\textwidth}
      \centering
      \includegraphics[page=3]{figures/AngleTransport}
      \caption{}
      \label{fig:transportC}
    \end{subfigure}
    \hfill
    \begin{subfigure}{0.24\textwidth}
      \centering
      \includegraphics[page=4]{figures/AngleTransport}
      \caption{}
      \label{fig:transportD}
    \end{subfigure}
    \caption{
    \textbf{(a)} An attachment vertex with one modulator edge $m$ and two attachment edges $f^1, g^1$ in $G_i^{\{1\}}$, together with a corresponding rotation $\mathcal E_i^{\{1\}}(x)$ and its three direct angles highlighted.
    \textbf{(b)} A $1$-consistent rotation $\mathcal E_i^{\{1, 2\}}(x)$, with the induced angle-assignment $\tphix^{2 \rightarrow 1}$ indicated with colors.
    \textbf{(c)} A $1$-consistent rotation $\mathcal E_i^{\{1, 2, 3\}}(x)$ that does not satisfy the transported angle-assignment $\phix^{3 \rightarrow 1}$, because $f^3$ does not lie in its designated angle of $\mathcal E_i^{\{1\}}(x)$.
    \textbf{(d)} A different rotation that is $2$-consistent and satisfies all transported angle-assignments.
    }
    \label{fig:transport}
\end{figure}

Let $\embU$ be a simultaneous rotation system of $\mc I$ such that $\embU$ is $1$-consistent with respect to $\mc K$.
Let $i \in [k]$, let $x \in \att_i(\mc K)$ and let $j \in [l]$.
For a direct angle $\alpha^j = (e^j, f^j)$ of $\emb_i^{\{j\}}(x)$ and $j' \in [l]$ we let $\alpha^{j'} = \iso^{j \rightarrow j'}(\alpha^j) \coloneqq (e^{j'}, f^{j'})$ denote the corresponding direct angle of $\emb_i^{\{j'\}}(x) = \iso^{j \rightarrow j'}(\emb_i^{\{j\}}(x))$.
We let $\tphix^{2 \rightarrow 1}$ be the direct angle-assignment from $A_i^2(x)$ to $\emb_i^{\{1\}}(x)$ induced by $\emb_i^{\{1,2\}}(x)$.
For $1 \leq j_1 < j_2 \leq l$, we let $\phix^{j_2 \rightarrow j_1} \coloneqq \iso^{(1,2) \rightarrow (j_1, j_2)} \circ \tphix^{2 \rightarrow 1} \circ \iso^{(j_1,j_2) \rightarrow (1, 2)}_{\vert A_i^{j_2}(x)}$ denote the corresponding direct angle-assignment from $A_i^{j_2}(x)$ to $\emb_i^{\{j_1\}}(x)$.

For each $S \subseteq [l]$ with $|S| \geq 2$ and each $1$-consistent simultaneous rotation system of~$\mc I^S$, we analogously define the angle-assignment $\phix^{j_2 \rightarrow j_1}$ for all $j_1 \in S$ and all $j_2 \in [l]$ with $j_1 < j_2$ (note that $j_2$ need not belong to $S$) based on the angle-assignment induced by the first two blobs of $\mathcal I^S$.
 If two $1$-consistent simultaneous rotation systems coincide on the first pair of blobs, they thus induce the same angle-assignments $\phix^{j_2 \rightarrow j_1}$.
Moreover, note that the choice of the reference pair is irrelevant if the simultaneous rotation system is $2$-consistent.

Since we only assume that $\embU$ is $1$-consistent, note that $\emb_i^{\{j_1,j_2\}}(x)$ does not necessarily satisfy $\phix^{j_2 \rightarrow j_1}$; see \Cref{fig:transport}.
However, we now show that satisfying these angle-assignments is necessary and sufficient for $2$-consistency.
Roughly speaking, by \Cref{lm:merges}, these angle-assignments uniquely define a combined rotation at $x$ for each pair of blobs.
As the angle-assignments coincide under isomorphism for all pairs of blobs by definition, the combined rotations thus also coincide under isomorphism for all pairs of blobs, implying $2$-consistency.
 
\begin{restatable}\restateref{lm:twoconsistent}{lemma}{lemmaTwoconsistent}
  \label{lm:twoconsistent}
  Let $\mc I = \Instance$ be an exclusive-connected instance of \SEFE.
  Let $\mc K$ be a blob-partition of size $l \geq 3$ of $\mc I$ whose modulator is extended and let $\embU$ be a simultaneous rotation system of $\mc I$.
  Then \embU is $2$-consistent with respect to $\mc K$ if and only if
  \begin{enumerate}
    \item \embU is $1$-consistent with respect to $\mc K$,
    \item $\embU^{[l-1]}$ is $2$-consistent with respect to $\mc K^{[l-1]}$, and
    \item for all $i \in [k]$ and each $x \in \att_i(\mc K)$, $\emb_i(x)$ satisfies the angle-assignments $\phix^{l \rightarrow 1}, \dots, \phix^{l \rightarrow l-1}$.
  \end{enumerate}
  
\end{restatable}
\begin{prooflater}{plemmaTwoconsistent}
  $(\Rightarrow)\colon$
        Assume \embU is $2$-consistent with respect to $\mc K$.
  We first show $1$-consistency, i.e., $\emb_i^{\{j\}} = \iso^{1 \rightarrow j}(\emb_i^{\{1\}})$ for all $j \in [l]$.
  For $j = 1$, the statement holds trivially.
  For $j > 1$, we have $\emb_i^{\{1,j\}} = \iso^{(1,2) \rightarrow (1,j)}(\emb_i^{\{1,2\}})$ by $2$-consistency and we thus get $\emb_i^{\{j\}} = \iso^{2 \rightarrow j}(\emb_i^{\{2\}})$. 
  Because $l \geq 3$, we also have $\emb_i^{\{2,3\}} = \iso^{(1,2) \rightarrow (2,3)}(\emb_i^{\{1,2\}})$ and thus $\emb_i^{\{2\}} = \iso^{1 \rightarrow 2}(\emb_i^{\{1\}})$.
  Hence $\embU$ is also $1$-consistent with respect to $\mc K$.

     Because $\embU$ is $2$-consistent with respect to $\mc K$, $\embU^{[l-1]}$ is $2$-consistent with respect to $\mc K^{[l-1]}$ by definition.
  Moreover, we have for all $j \in [l-1]$ and all $x \in \att_i(\mc K)$ that $\emb_i^{\{j, l\}}(x) = \iso^{(1,2)\rightarrow(j,l)}(\emb_i^{\{1,2\}}(x))$.
  Hence the two rotations coincide under $\iso^{(1,2)\rightarrow(j,l)}$.
  Since $\emb_i^{\{1,2\}}(x)$ satisfies $\tphix^{2 \rightarrow 1}$ by definition, $\emb_i^{\{j, l\}}(x)$ thus satisfies $\phix^{l \rightarrow j}$, which also coincides with $\tphix^{2 \rightarrow 1}$ under $\iso^{(1,2)\rightarrow(j,l)}$ by definition.

  $(\Leftarrow)\colon$
     Let $i \in [k]$ and let $1 \leq j_1 < j_2 \leq l$.
  To show $2$-consistency, we show that $\emb_i^{\{j_1,j_2\}}(\iso^{(1,2)\rightarrow(j_1,j_2)}(x)) = \iso^{(1,2)\rightarrow(j_1,j_2)}(\emb_i^{\{1,2\}}(x))$ for all $x \in V(G_i^{\{1,2\}})$.

  If $x \in B^1 \cup (X \setminus \att_i(\mc K))$, all edges incident to $x$ in $G_i$ belong to $G_i^{\{1\}}$ and all edges incident to $\iso^{1 \rightarrow j_1}(x)$ belong to $G_i^{\{j_1\}}$.
  Hence $1$-consistency yields $\iso^{(1,2)\rightarrow(j_1,j_2)}(\emb_i^{\{1,2\}}(x)) = \iso^{(1,2)\rightarrow(j_1,j_2)}(\emb_i^{\{1\}}(x)) = 
  \iso^{1\rightarrow j_1}(\emb_i^{\{1\}}(x)) = \emb_i^{\{j_1\}}(\iso^{1 \rightarrow j_1}(x)) = \emb_i^{\{j_1,j_2\}}(\iso^{(1,2)\rightarrow(j_1,j_2)}(x))$ as desired.
  The case $x \in B^2$ is analogous.

  Now consider the case $x \in \att_i(\mc K)$.
  Because $\att_i(\mc K) \subseteq X$, we have ${\iso^{(1,2)\rightarrow(j_1,j_2)}(x) = x}$.
        If $j_2 < l$, the statement follows immediately from the $2$-consistency of $\embU^{[l-1]}$ with respect to $\mc K^{[l-1]}$.
  Consider the case $j_2 = l$.
  Note that we have that $\emb_i^{\{j_1,l\}}(x)$ satisfies $\phix^{l \rightarrow j_1}$ by assumption and $\emb_i^{\{1,2\}}(x)$ satisfies $\tphix^{2 \rightarrow 1}$ by definition.
  By \Cref{lm:merges}, $\emb_i^{\{j_1,l\}}(x)$ is uniquely defined by $\emb_i^{\{j_1\}}(x)$, $\emb_i^{\{l\}}(x)$, and $\phix^{l \rightarrow j_1}$.
  Analogously, $\emb_i^{\{1,2\}}(x)$ is uniquely defined by $\emb_i^{\{1\}}(x)$, $\emb_i^{\{2\}}(x)$, and $\tphix^{2 \rightarrow 1}$.
  Due to $1$-consistency, we have that $\emb_i^{\{j_1\}}(x) = \iso^{(1,2) \rightarrow (j_1, l)}(\emb_i^{\{1\}}(x))$ and $\emb_i^{\{l\}}(x) = \iso^{(1,2) \rightarrow (j_1, l)}(\emb_i^{\{2\}}(x))$.
  Moreover, since we have $\phix^{l \rightarrow j_1} = \iso^{(1,2) \rightarrow (j_1, l)} \circ \tphix^{2 \rightarrow 1} \circ \iso^{(j_1,l) \rightarrow (1, 2)}_{\vert A_i^{l}(x)}$ by definition, the three ingredients that uniquely determine $\emb_i^{\{j_1,l\}}(x)$ and $\emb_i^{\{1,2\}}(x)$, respectively, coincide under $\iso^{(1,2) \rightarrow (j_1,l)}$.
  Hence we have $\emb_i^{\{j_1,l\}}(x) = \iso^{(1,2) \rightarrow (j_1,l)}(\emb_i^{\{1,2\}}(x))$, which concludes the proof.
\end{prooflater}

We remark that $2$-consistency does not imply $1$-consistency for $l = 2$.

By \Cref{lm:twoconsistent}, to construct the desired embedding $\emb_i$ of $G_i$, we thus have to merge the rotation systems $\emb_i^{[l-1]} = \embP_i$ of $G_i^{[l-1]}$ and $\emb_i^{\{l\}} = \kappa^{1\rightarrow l}({\embP_i}^{\{1\}})$ of $G_i^{\{l\}}$ at every modulator vertex $x \in \att_i(\mc K)$ in such a way that the induced angle-assignment from $A^l_i(x)$ to $\emb_i^{[l-1]}(x)$ refines $\phix^{l \rightarrow 1}, \dots, \phix^{l \rightarrow l-1}$.
 Our next goal is to show that such an angle-assignment exists.
While we only need this for $l \geq |X| + 4$ to construct the desired embedding, we also show that such an angle-assignment is unique for $l \geq 3$, if it exists.
We will use this later to show that $2$-consistency implies $m$-consistency for all $m \in [l]$.

\begin{figure}
    \centering
    \begin{subfigure}{0.24\textwidth}
      \centering
      \includegraphics[page=1]{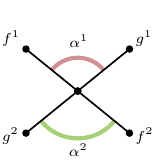}
      \caption{}
      \label{fig:intersectionA}
    \end{subfigure}
    \hfill
    \begin{subfigure}{0.24\textwidth}
      \centering
      \includegraphics[page=2]{figures/angleIntersection}
      \caption{}
      \label{fig:intersectionB}
    \end{subfigure}
    \hfill
    \begin{subfigure}{0.24\textwidth}
      \centering
      \includegraphics[page=3]{figures/angleIntersection}
      \caption{}
      \label{fig:intersectionC}
    \end{subfigure}
    \hfill
    \begin{subfigure}{0.24\textwidth}
      \centering
      \includegraphics[page=4]{figures/angleIntersection}
      \caption{}
      \label{fig:intersectionD}
    \end{subfigure}
    \caption{Four configurations for two angles $\alpha^1 = (f^1, g^1)$ and $\alpha^2 = (f^2, g^2)$.
    In \textbf{(b)-(c)}, the blue angle highlights the unique intersection of $\alpha^1, \alpha^2, \dots, \alpha^{l-1}$.}
    \label{fig:angleIntersection}
\end{figure}
  
\begin{restatable}\restateref{lm:angleAssignment}{lemma}{lemmaAngleAssignment}
  \label{lm:angleAssignment}
  Let $\mc I = \Instance$ be an exclusive-connected instance of \SEFE.
  Let $\mc K$ be a blob-partition of size $l \geq 3$ of $\mc I$ whose modulator is extended and let $\embPU$ be a simultaneous rotation system of $\mc I^{[l-1]}$ that is $1$- and $2$-consistent with respect to $\mc K^{[l-1]}$.
  Let $i \in [k]$ and $x \in \att_i(\mc K)$.
  Then we have the following.
  \begin{enumerate}
      \item If $l \geq 3$, there exists at most one direct angle-assignment $\phix$ from $A^{l}_i(x)$ to $\embP_i(x)$ that refines $\phix^{l \rightarrow 1}, \dots, \phix^{l \rightarrow l-1}$.
      \item If $l \geq 4$, such an angle-assignment $\phix$ exists. Moreover, $\phix$ is a tidy refinement of $\phix^{l \rightarrow 1}, \dots, \phix^{l \rightarrow l-1}$.
                \end{enumerate}
\end{restatable}
\begin{proofsketch}
    For each edge $e^{l} \in A_i^{l}(x)$, we consider the angles $\alpha^1 = \phix^{l \rightarrow 1}(e^l), \dots, \alpha^{l-1} = \phix^{l \rightarrow l-1}(e^l)$ in $\embP_i(x)$ to show that there exists at most one direct angle that lies inside all of them (and it exists for $l \geq 4$).
    To this end, we first consider the two angles $\alpha^1$ and $\alpha^2$ as intervals in a cyclic order and distinguish whether they (a) are disjoint or touch only at their endpoints, (b) are nested, (c) overlap once, or (d) overlap twice; see \Cref{fig:angleIntersection}.

    While the last case yields two disjoint intervals that lie inside of $\alpha^1$ and $\alpha^2$, we can use the fact that $X$ is an extended modulator to show that this case cannot occur. As $\alpha^1$ and $\alpha^2$ are direct angles in $\embP_i^{\{1\}}$ and $\embP_i^{\{2\}}$, respectively, the corresponding modulator edge cannot lie inside either of them. 
    But in case (d), every edge lies in one of the two angles, a contradiction.

    Due to $2$-consistency, the remaining angles $\alpha^3, \dots, \alpha^{l-1}$ must adhere to the same pairwise pattern as $\alpha^1$ and $\alpha^2$.
    In the first three cases, we thus find at most one direct angle that lies inside $\alpha^1, \dots \alpha^{l-1}$, which gives the first statement.

    For $l \geq 4$, the definition of $\phix^{l \rightarrow 1}$ via $\tphix^{2 \rightarrow 1}$ tells us that $\alpha^1$ contains the edge $e^2$.
    Due to $2$-consistency, the edge $e^3$ thus lies inside both $\alpha^1$ and $\alpha^2$ in $\embP_i(x)$.
    Hence $\alpha^1$ and $\alpha^2$ cannot have disjoint interiors and thus case (a) cannot occur and we always find the desired direct angle that lies inside $\alpha^1, \dots, \alpha^{l-1}$.
    Note that we can only use this argument if $l - 1 \geq 3$.
    We can additionally confirm the tidiness of the refinements, which shows the second statement.
\end{proofsketch}
\begin{prooflater}{plemmaAngleAssignment}
  Let $e^{l} \in A_i^{l}(x)$. 
  We consider the angles $\phix^{l \rightarrow 1}(e^l), \dots \phix^{l \rightarrow l-1}(e^l)$ in $\embP_i(x)$.
  Let $j \in [l-1]$ and let $\alpha^j =(f^j, {g}^j) \coloneqq \phix^{l\rightarrow j}(e^l)$.
  Recall that
  \[ \alpha^j = \phix^{l \rightarrow j}(e^l) = \iso^{(1,2) \rightarrow (j,l)}(\tphix^{2 \rightarrow 1}(\iso^{(j,l) \rightarrow (1,2)}(e^l))) = \iso^{(1,2) \rightarrow (j,l)}(\tphix^{2 \rightarrow 1}(e^2)) = \iso^{1 \rightarrow j}(\tphix^{2 \rightarrow 1}(e^2)), \]
  the different $\alpha^j$ thus coincide under isomorphism.
  Note that, while $\alpha^j$ is a direct angle of $\embP^{\{j\}}_i(x)$, it is not necessarily a direct angle of $\embP_i(x)$.
  For every $j' \in [l-1]\setminus \{j\}$, recall that $f^{j'} = \iso^{j\rightarrow j'}(f^j)$ and thus $f^{j'} = f^j$ if and only if $f^j$ is a modulator edge.
  The analogous holds for ${g}^j$.

  Consider the angles $\alpha^1 = \phix^{l\rightarrow 1}(e^l)$ and $\alpha^2 = \phix^{l\rightarrow 2}(e^l)$ in $\embP_i(x)$.  
       We distinguish cases based on whether $\alpha^1$ and $\alpha^2$, which are intervals in a cyclic order, (a) are disjoint or touch only at their endpoints, (b) are nested, (c) overlap once, or (d) overlap twice; see \Cref{fig:angleIntersection}.
  
  In case (a), there exists no angle that lies inside $\alpha^1 = \phix^{l \rightarrow 1}(e^l)$ and $\alpha^2 = \phix^{l \rightarrow 2}(e^l)$ and thus $ \phix^{l \rightarrow 1}$ and $ \phix^{l \rightarrow 2}$ do not have a common refinement.
  
  In case (b), we assume without loss of generality that $\alpha^2$ lies inside $\alpha^1$, i.e., their corresponding edges appear in the clockwise order $f^1, f^2, g^2, g^1$, where possibly $f^1 = f^2$ or $g^1 = g^2$.
  Since the corresponding angles $\alpha^3, \dots, \alpha^{l-1}$ must adhere to the same pairwise pattern due to $2$-consistency, their edges appear in the clockwise order $f^1, f^2, \dots, f^{l-1}, g^{l-1}, \dots g^2, g^1$ (where possibly $f^1 = f^2 = \dots = f^{l-1}$ or $g^1 = g^2 = \dots = g^{l-1}$).
  We claim that the angle $\alpha_i(e^{l}) \coloneqq \alpha^{l-1} = (f^{l-1}, g^{l-1})$ is direct.
  Assume for the sake of contradiction that there is an edge $h$ that lies between $f^{l-1}$ and $g^{l-1}$ in $\embP^{[l-1]}_i(x)$.
  Note that $h$ is contained in each angle $\alpha^1, \dots, \alpha^{l-1}$.
  But since $h$ belongs to graph ${G_i}^{\{j\}}$ for some $j \in \{1, \dots, l-1\}$ (or to all of them, if $h$ is a modulator edge), this contradicts the fact that $\alpha^{j}$ is a direct angle in $\embP_i^{\{j\}}(x)$.
  Thus $\alpha_i(e^{l})$ is a direct angle in $\embP^{[l-1]}_i(x)$.
  Moreover, observe that it is the only direct angle that lies inside $\alpha^1, \dots, \alpha^{l-1}$.

  Note that $f^1 = f^2$ or $g^1 = g^2$ implies case (b); we thus assume for the remaining two cases that $f^1 \neq f^2$ and $g^1 \neq g^2$.
  Hence none of these edges is a modulator edge, i.e., $f^1, g^1 \in A_i^1(x)$ and $f^2, g^2 \in A_i^2(x)$ and thus also $f^1 \neq g^2$ and $f^2 \neq g^1$.
  All four edges are therefore distinct.

  In case (c), assume without loss of generality that their edges appear in the circular order $f^1, f^2, g^1, g^2$ in $\embP^{[l-1]}_i(x)$.
  As in case (b), $2$-consistency guarantees that the corresponding angles $\alpha^3, \dots, \alpha^{l-1}$ must adhere to the same pattern and thus their edges appear in the clockwise order $f^1, f^2, \dots, f^{l-1}, g^1, g^2, \dots g^{l-1}$ in $\embP^{[l-1]}_i(x)$.
     Using the same argument as in case (b), we obtain that the angle $\alpha_i(e^{l}) \coloneqq (f^{l-1}, g^1)$ is the unique direct angle of $\embP^{[l-1]}_i(x)$ that lies inside $\alpha^1, \dots, \alpha^{l-1}$.

  Finally, consider case (d). 
  Assume without loss of generality that the edges of $\alpha^1$ and $\alpha^2$ appear in the circular order  $f^1, g^2, f^2, g^1$ in $\embP^{[l-1]}_i(x)$.
  This implies that every element of $\embP^{[l-1]}_i(x)$ lies inside $\alpha^1$ or $\alpha^2$.
  Since the modulator $X$ is extended, $x$ is incident to a modulator edge $m \in X_i(x)$.
  Hence $\embP^{\{1\}}_i(x)$ and $\embP^{\{2\}}_i(x)$ both contain $m$.
  Since $\alpha^1$ and $\alpha^2$ are direct angles in $\embP^{\{1\}}_i(x)$ and $\embP^{\{2\}}_i(x)$, $m$ cannot lie inside $\alpha^1$ or $\alpha^2$ in $\embP^{[l-1]}_i(x)$, a contradiction.

  \medskip

  Hence case (d) cannot occur.  
  In each of the other cases, the direct angle $\phix(e^l) = \alpha_i(e^l)$ is uniquely determined, if it exists.
  Hence $\phix$ is unique.
  This shows the first statement.
  Moreover, if this angle exists for each edge $e^l \in A_i^l(x)$, i.e., case (a) never occurs, then $\phix$ is the desired direct angle-assignment.

  Now assume $l \geq 4$.
  Recall that $\alpha^1 = \kappa^{1\rightarrow 1}(\tphix^{2 \rightarrow 1}(e^2)) = \tphix^{2 \rightarrow 1}(e^2)$ and thus $e^2$ lies inside $\alpha^1$ in $\embP_i(x)$.
  Since $\embPU$ is $2$-consistent, it follows that $e^3$ lies inside $\alpha^1$ and $\alpha^2$ in $\embP_i(x)$.
  This excludes case (a) and therefore proves that $\phix$ exists.
  It thus only remains to prove the tidiness of the refinements.
  Let $j \in [l-1]$ and let $e_1^l, e_2^l \in A^{l}_i(x)$.
  Assume that $\phix^{l \rightarrow j}(e_1^l) = \phix^{l \rightarrow j}(e_2^l)$ holds.
  By the definition of $\phix^{l \rightarrow j}$, we obtain $\tphix^{2 \rightarrow 1}(e_1^2) = \tphix^{2 \rightarrow 1}(e_2^2)$ and thus also $\phix^{l \rightarrow j'}(e_1^l) = \phix^{l \rightarrow j'}(e_2^l)$ for all $j' \in [l-1]$.
  Since $\phix(e_1^l)$ is the unique direct angle that lies inside $\phix^{l \rightarrow 1}(e_1^l), \dots, \phix^{l \rightarrow l-1}(e_1^l)$ and $\phix(e_2^l)$ is the unique direct angle that lies inside $\phix^{l \rightarrow 1}(e_2^l), \dots, \phix^{l \rightarrow l-1}(e_2^l)$, it follows that $\phix(e_1^l) = \phix(e_2^l)$.
  Thus $\phix$ is a tidy refinement of $\phix^{l \rightarrow j}$.
  This proves the second statement.
\end{prooflater}

With this result, we can now show that $\embPU$ can be extended to the desired $2$-consistent simultaneous rotation system $\embU$ of \GU.
For each $i \in [k]$ and each $x \in \att_i(\mathcal K)$, we use the uniquely defined tidy refinement $\phix$ from \Cref{lm:angleAssignment} (which exists, as $l \geq 4)$ to merge $\emb_i^{[l-1]}(x) = {\embP_i}^{[l-1]}(x)$ and $\emb_i^{\{l\}}(x) = \iso^{1 \rightarrow l}(\embP_i^{\{1\}})(x)$ into one circular order $\emb_i(x)$.
To this end, we use the embedding induced by the first two blobs as a witness (under isomorphism) that there exists a circular ordering that satisfies $\phix^{l \rightarrow 1}$.
We can then use the fact that $\phix$ is a tidy refinement of $\phix^{l \rightarrow 1}$ to show that the desired circular ordering $\emb_i(x)$ exists.

\begin{restatable}\restateref{lm:twoConsEmb}{lemma}{lemmaTwoConsEmb}
\label{lm:twoConsEmb}
  Let $\mc I = \Instance$ be an exclusive-connected instance of \SEFE.
  Let $\mc K$ be a blob-partition of size $l \geq 4$ of $\mc I$ whose modulator is extended.
  Moreover, let $\embPU$ be a simultaneous rotation system of $\mc I^{[l-1]}$ such that $\embPU$ is $2$-consistent with respect to $\mc K^{[l-1]}$.
  
  There exists a simultaneous rotation system $\embU$ of $\mc I$ that coincides with $\embPU$ on $\GU^{[l-1]}$ and such that \embU is $2$-consistent with respect to $\mc K$.
\end{restatable}
\begin{prooflater}{plemmaTwoConsEmb}
  For each $i \in [k]$, we construct a rotation system $\emb_i$ of $G_i$ as follows.   
  Recall the characterization of $2$-consistency from \Cref{lm:twoconsistent}.
  In particular, $\embU$  must be $1$-consistent with respect to $\mc K$ and thus it is necessary that $\emb_i^{\{l\}} = \iso^{1 \rightarrow l}(\embP_i^{\{1\}})$.
  Moreover, since $\embU$ must coincide with $\embPU$ on $\GU^{[l-1]}$, it is also necessary that $\emb_i^{[l-1]} = {\embP_i}^{[l-1]}$.
  We combine these two embeddings to a single rotation system $\emb_i$ as follows.
  Note that $\emb_i^{\{l\}}$ and $\emb_i^{[l-1]}$ already uniquely define $\emb_i(x)$ for all $x \notin \att_i(\mc K)$.
  It remains to specify how to merge the circular orderings $\emb_i^{[l-1]}(x)$ and $\emb_i^{\{l\}}(x)$ at each vertex $x \in \att_i(\mc K)$ into a circular ordering $\emb_i(x)$ in such a way that $2$-consistency is preserved.
  
  By \Cref{lm:twoconsistent}, it is necessary that $\emb_i(x)$  satisfies the angle-assignments $\phix^{l \rightarrow 1}, \dots \phix^{l \rightarrow l-1}$.
  In other words, the direct angle-assignment from $A_i^l(x)$ to $\emb_i^{[l-1]}(x)$ induced by $\emb_i(x)$ must be a refinement of $\phix^{l \rightarrow 1}, \dots \phix^{l \rightarrow l-1}$.
  By \Cref{lm:angleAssignment}, there exists a unique direct angle-assignment $\phix$ that satisfies this property.
  Since the modulator is extended, by \Cref{lm:merges} there is at most one circular ordering $\emb_i(x)$ of the edges incident to $x$ in $G_i$ that extends both $\emb_i^{[l-1]}(x) = \embP_i^{[l-1]}(x)$ and $\emb_i^{\{l\}}(x) = \iso^{1 \rightarrow l}(\embP_i^{\{1\}}(x))$ and additionally satisfies $\phix$.
  We still have to show that this ordering exists.

  Since the resulting circular ordering must satisfy $\phix$, we must place the missing edges $A_i^l(x)$ into their designated direct angles of $\emb_i^{[l-1]}(x)$.
  However, we have to ensure that we do this consistently with $\emb_i^{\{l\}}(x)$.
  To this end, we use the circular ordering $\tau^{\{1,l\}} \coloneqq  \iso^{(1,2) \rightarrow (1,l)}(\embP_i^{\{1,2\}}(x))$ induced for blobs $1$ and $l$ as guidance.
  Note that $\tau^{\{1,l\}}$ satisfies $\phix^{l \rightarrow 1}$.

  Let $e_1, e_2 \in A_i^l(x)$.
  Because $\phix$ refines $\phix^{l \rightarrow 1}$ and each direct angle of $\emb_i^{[l-1]}(x)$ lies inside exactly one direct angle of $\emb_i^{\{1\}}(x)$, we have $\phix(e_1) = \phix(e_2) \Rightarrow \phix^{l \rightarrow 1}(e_1) = \phix^{l \rightarrow 1}(e_2)$.
  Since the refinement is tidy, the converse also holds and thus $\phix(e_1) = \phix(e_2) \Leftrightarrow \phix^{l \rightarrow 1}(e_1) = \phix^{l \rightarrow 1}(e_2)$.
  Hence $\phix$ and $\phix^{l \rightarrow 1}$ induce a unique bijection $f \colon \Ima(\phix) \to \Ima(\phix^{l \rightarrow 1})$ with $\phix^{-1}(\alpha) = (\phix^{l \rightarrow 1})^{-1}(f(\alpha))$ for all $\alpha \in \Ima(\phix)$.
  
  We now construct the desired circular ordering.
  For an angle $\alpha \in \Ima(\phix)$, let $\pi_\alpha$ be the linear order in which the edges $\phix^{-1}(\alpha) =  (\phix^{l \rightarrow 1})^{-1}(f(\alpha))$ appear inside $f(\alpha)$ in $\tau^{\{1,l\}}$.
  Let $\sigma_{i,x}$ be the circular ordering obtained from $\emb_i^{[l-1]}(x)$ by inserting, for each $\alpha \in \Ima(\phix)$, the elements of $\phix^{-1}(\alpha) =  (\phix^{l \rightarrow 1})^{-1}(f(\alpha))$ inside $\alpha$ in the order $\pi_\alpha$.
  Because $\alpha$ lies inside $f(\alpha)$ and $f$ is a bijection, this construction guarantees that the restriction of $\sigma_{i,x}$ to the edges of $G_i^{\{1,l\}}$ equals $\tau^{\{1,l\}}$.
  Since $\tau^{\{1,l\}} = \iso^{(1,2) \rightarrow (1,l)}(\embP_i^{\{1,2\}}(x))$, restricting $\sigma_{i,x}$ to the edges of $G_i^{\{l\}}$ thus yields $\iso^{2 \rightarrow l}(\embP_i^{\{2\}}(x))$, which equals $\iso^{1 \rightarrow l}(\embP_i^{\{1\}}(x)) = \emb_i^{\{l\}}(x)$ by $1$-consistency of~$\embPU$.
  Moreover, $\sigma_{i,x}$ extends $\emb_i^{[l-1]}(x)$ and satisfies $\phix$ (and thus also $\phix^{l \rightarrow 1}, \dots \phix^{l \rightarrow l-1}$) by construction.
  
  Now set $\emb_i(x) \coloneqq \sigma_{i,x}$ for all $i \in [k]$ and $x \in \att_i(\mathcal K)$.
  Then we have that $\embU$ is $2$-consistent by \Cref{lm:twoconsistent}: $\embU$ is $1$-consistent and satisfies the corresponding angle-assignments by the above, and $\embU^{[l-1]} = \embPU$ is $2$-consistent by assumption.
  This concludes the proof.
\end{prooflater}

To simplify the final proof, we first show that $2$-consistency implies $m$-consistency for all $m \in [l]$.
We do this inductively on $m$ and in each step show that any $m$ blobs induce the same rotation in each input graph.
To this end, we again use the fact that the angle-assignments from the last of these blobs to the rotation of the other $m-1$ blobs have a unique common refinement by \Cref{lm:angleAssignment}.
This refinement is identical (under isomorphism) for any $m$ blobs and thus \Cref{lm:merges}, together with ($m\!-\!1$)-consistency, yields that the combined rotation also coincides under isomorphism, implying $m$-consistency.

\begin{restatable}\restateref{lm:m-consistent}{lemma}{lemmaMConsistent}
\label{lm:m-consistent}
  Let $\mc I = \Instance$ be an exclusive-connected instance of \SEFE.
  Let $\mc K$ be a blob-partition of size $l \geq 3$ of $\mc I$ whose modulator is extended and let $\embU$ be a simultaneous rotation system of $\mc I$.
  If $\embU$ is $2$-consistent with respect to $\mc K$, then $\embU$ is $m$-consistent with respect to $\mc K$ for all $m \in [l]$.
\end{restatable}
\begin{prooflater}{plemmaMConsistent}
  The proof is by induction on $m$.
     Note that $1$-consistency follows from $2$-consistency by \Cref{lm:twoconsistent}.
  Now let $m \geq 3$ and assume that \embU is $(m\!-\!1)$-consistent with respect to $\mc K$.
  We show that  \embU is also $m$-consistent with respect to $\mc K$.
  To this end, let $(j_1, \dots, j_m)$ be a strictly increasing sequence over $[l]$ and let $\iso \coloneqq \iso^{(1, \dots, m) \rightarrow (j_1, \dots, j_m)}$.
  Let $S = \{j_1, \dots, j_m\}$ and $S^- = S \setminus \{j_m\}$.
  We want to show for each $i \in [k]$ and each $x \in V(G_i^{[m]})$ that $\emb_i^{S}(\iso(x)) = \iso(\emb_i^{[m]}(x))$.
  As above, this holds for $x \in V(G_i^{[m]}) \setminus \att_i(\mc K)$ due to $1$-consistency.
  
  Consider $x \in \att_i(\mc K)$ and recall that $x = \iso(x)$ in this case.
  Let $\phix^{j_m \rightarrow S^-}$ denote the direct angle-assignment from $A_i^{j_m}(x)$ to $\emb_i^{S^-}(x)$ induced by $\emb_i^S(x)$.
  Because $\embU$ is $2$-consistent with respect to $\mc K$, we have that $\emb_i^S(x)$ satisfies $\phix^{j_m \rightarrow j_1}, \dots, \phix^{j_m \rightarrow j_{m-1}}$ by \Cref{lm:twoconsistent}\footnote{When applying Lemma~\ref{lm:twoconsistent} or~\ref{lm:angleAssignment} to a subinstance $\mathcal I^S$, we implicitly re-index the blobs of $\mathcal K^S$ order-preservingly to $[|S|]$. 
  Due to the $2$-consistency of $\embU$, the definition of $\phix^{j_m \rightarrow j_1}, \dots, \phix^{j_m \rightarrow j_{m-1}}$ based on the angle-assignment $\tphix^{2 \rightarrow 1}$ induced by the first two blobs still coincides in $\mathcal I^S$ and $\mathcal I$.
  } and thus $\phix^{j_m \rightarrow S^-}$ refines $\phix^{j_m \rightarrow j_1}, \dots, \phix^{j_m \rightarrow j_{m-1}}$.
                 
  We analogously obtain that the direct angle-assignment $\phix^{m \rightarrow [m-1]}$ from $A_i^m(x)$ to $\emb_i^{[m-1]}$ induced by $\emb_i^{[m]}(x)$ refines $\phix^{m \rightarrow 1}, \dots, \phix^{m \rightarrow m-1}$.
  Applying \Cref{lm:angleAssignment} to the two subinstances, $\phix^{j_m \rightarrow S^-}$ and $\phix^{m \rightarrow [m-1]}$ are the only such angle-assignments that satisfy this refinement property.
  Since the angle-assignments $\phix^{j_m \rightarrow j_1}, \dots, \phix^{j_m \rightarrow j_{m-1}}$ coincide with the angle-assignments  $\phix^{m \rightarrow 1}, \dots, \phix^{m \rightarrow m-1}$ under $\iso$ by definition, $\phix^{j_m \rightarrow S^-}$ and $\phix^{m \rightarrow [m-1]}$ also coincide under $\iso$.
  Moreover, we have $\emb_i^{S^-}(x) = \iso(\emb_i^{[m-1]}(x))$ and $\emb_i^{\{j_{m}\}}(x) = \iso(\emb_i^{\{m\}}(x))$ by $(m\!-\!1)$-consistency.
  By \Cref{lm:merges}, $\emb_i^{S}(x)$ is uniquely determined by $\emb_i^{S^-}(x)$, $\emb_i^{\{j_{m}\}}(x)$, and $\phix^{j_m \rightarrow S^-}$ and $\emb_i^{[m]}(x)$ is uniquely determined by $\emb_i^{[m-1]}(x)$, $\emb_i^{\{m\}}(x)$, and $\phix^{m \rightarrow [m-1]}$.
  Hence the three ingredients that uniquely define $\emb_i^{S}(x)$ and $\emb_i^{[m]}(x)$, respectively, coincide under $\iso$.
  We thus have  $\emb_i^{S}(x) = \iso(\emb_i^{[m]}(x))$, which concludes the proof.
\end{prooflater}

Using \Cref{lm:m-consistent}, we can now show that the $2$-consistent rotation system $\embU$ from \Cref{lm:twoConsEmb} is a simultaneous embedding.
We use the fact that non-planarity (\Cref{lm:planarityCharacterization}) and inconsistent embeddings are witnessed by a small number of blobs.
Due to \Cref{lm:twoConsEmb,lm:m-consistent}, every proper subset of blobs induces a simultaneous embedding, and we can conclude that the same also holds for the entire instance.

\begin{lemma}
\label{lm:simEmb}
  Let $\mc I = \Instance$ be an exclusive-connected instance of \SEFE.
  Let~$\mc K$ be a blob-partition of $\mc I$ with an extended modulator $X$ and size $l \geq |X| + 3$.
  Moreover, let $\embU$ be a simultaneous rotation system of $\mc I$ such that $\embU$ is $2$-consistent with respect to~$\mc K$.
  If $\embU^{[l-1]}$ is a simultaneous embedding of $\mc I^{[l-1]}$, then  $\embU$ is a simultaneous embedding of~$\mc I$.
\end{lemma}
\begin{proof}
              We start with the following observation.
    \begin{claim}
    \label{cl:subsets}
        For each non-empty $S \subsetneq [l]$, $\embU^S$ is a simultaneous embedding of $\mc I^S$.
    \end{claim}
    \begin{claimproof}
        Since \embU is $2$-consistent, it is also $|S|$-consistent by \Cref{lm:m-consistent}.
        Because $[|S|] \subseteq [l-1]$, we have that $\embU^{[|S|]}$ is a simultaneous embedding by assumption and thus $|S|$-consistency immediately implies that $\embU^S$ is a simultaneous embedding of $\mc I^S$.
    \end{claimproof}

    In order to show that $\embU$ is a simultaneous embedding, we use the fact that obstructions are contained in a small subset of blobs.   
    \begin{restatable}\restateref{cl:smallObstruction}{claim}{claimSmallObstruction}
    \label{cl:smallObstruction}
        Let $H$ be a subgraph of $\GU$ that consists of a simple cycle $C$ that either has no bridge or a single bridge that is a path or a cycle.
        Then there exists a subset $S \subseteq [l]$ with $|S| \leq |X| + 1$ such that $H$ is a subgraph of $\GU^S$.
    \end{restatable}
    \begin{prooflater}{pclaimSmallObstruction}
        Let $\mc M$ denote the connected components of $H - X$ and recall that each such component is entirely contained in a single blob.
        Our goal is to show that $|\mc M| \leq |X| + 1$.
        If $H$ contains no vertex of $X$, we have $|\mc M| = 1$ and thus the statement holds.
        Note that $H$ is 2-edge-connected and thus each component of $\mc M$ is connected to $X$ with at least two edges.
        Hence we obtain $|\mc M| \leq \frac{\sum_{x \in X}\deg_H(x)}{2}$.
        If $C$ has no bridge, all vertices have degree 2 and thus $\sum_{x \in X}\deg_H(x) \leq 2|X|$.
        Otherwise, only the two fixtures of its unique bridge cause degrees that exceed 2 and each fixture contributes one additional unit.
        Hence $\sum_{x \in X}\deg_H(x) \leq 2|X| + 2$.
        In both cases, we obtain $|\mc M| \leq |X| + 1$ as desired.
        Since each connected component in $\mc M$ is entirely contained in a single blob, there consequently exists a subset $S \subseteq [l]$ with $|S| \leq |X| + 1$ such that $\GU^S$ contains all vertices of $\mc M$.
        If $\mc M = \emptyset$, we choose $S = \{1\}$.
        Since all remaining vertices of $H$ belong to $X$, $H$ is entirely contained in $\GU^S$.
    \end{prooflater}
  
    Using this, we now first show that each individual embedding is planar.
    Let $i \in [k]$ and assume for the sake of contradiction that $\emb_i$ is not planar.
    By \Cref{lm:planarityCharacterization}, there exists a simple cycle $C$ of $G_i$ and a bridge $B$ of $C$ with fixtures $a_1 = (x_1, e_1)$ and $a_2 = (x_2, e_2)$ such that $\emb_i$ places $a_1$ and $a_2$ on different sides of $C$.
    Let $P$ be an arbitrary simple $x_1x_2$-path in $B$ (or a simple cycle if $x_1=x_2$) that starts with $e_1$ and ends with $e_2$.
    Observe that $P$ and $C$ only share vertices $x_1$ and $x_2$ and let $H = C \cup P$.
    By \Cref{cl:smallObstruction}, we find a subset $S \subseteq [l]$ with $|S| \leq |X| + 1$ such that $\GU^S$ (and thus $G_i^S)$ contains $H$.
    Let $B'$ denote the unique bridge of $C$ in $G_i^S$ that contains $P$.
    Note that $a_1$ and $a_2$ are fixtures of $B'$ that lie on different sides of $C$ in $\emb_i^S$.
    Because $G_i^S$ is connected by \Cref{lm:connectedSubinstances}, $\emb_i^S$ is therefore non-planar by \Cref{lm:planarityCharacterization}.
    But since $|S| \leq |X| + 1$, we have $S \subsetneq [l]$ and we thus obtain a contradiction to \Cref{cl:subsets}.
    Therefore, $\emb_i$ is a planar embedding of $G_i$.
    
    It remains to show that the planar embeddings $\embU = (\emb_1, \dots, \emb_k)$ are consistent.
    We first show that, for every vertex $v$ of $\GU$, the order of shared edges incident to $v$ is consistent in $\emb_1, \dots, \emb_k$.
    Let $e_1, e_2, e_3$ be three shared edges incident to $v$ in $\GU$.
         Let $j_1, j_2, j_3 \in [l]$ such that $e_1, e_2, e_3 \in E(\GU^{\{j_1, j_2, j_3\}})$.
    By \Cref{cl:subsets}, we have that $\embU^{\{j_1, j_2, j_3\}}$ is a simultaneous embedding and thus, for each $i_1, i_2 \in [k]$, $e_1, e_2, e_3$ appear in the same circular order in $\emb^{\{j_1, j_2, j_3\}}_{i_1}(v)$ and $\emb^{\{j_1, j_2, j_3\}}_{i_2}(v)$.
    Hence the same also holds for $\emb_{i_1}(v)$ and $\emb_{i_2}(v)$.
    Therefore, the embeddings of $\embU$ have consistent orderings for shared edges.

    Finally, we show that $\emb_1, \dots, \emb_k$ have consistent relative positions in the shared graph.
    Let $H, H'$ denote distinct connected components of $G_\cap$, let $C$ be a simple cycle of $H$ and let $v$ be a vertex of $H'$.
    Let $i_1, i_2 \in [k]$.
    Using \Cref{cl:smallObstruction} and accounting for the fact that $v$ may lie in an additional blob, we find a subset $S \subsetneq [l]$ with $|S| \leq |X| + 2 < l$ such that $C$ and $v$ are contained in~$\GU^S$, and thus also in $G_\cap^S$.
         Since $\embU^S$ is a simultaneous embedding of $\mc I^S$ by \Cref{cl:subsets}, $v$ lies on the same side of $C$ in $\emb_{i_1}^S$ and $\emb_{i_2}^S$.
    Thus $v$ also lies on the same side of $C$ in $\emb_{i_1}$ and $\emb_{i_2}$.
    Hence $v$ lies on the same side of every simple cycle of $H$ in $\emb_{i_1}$ and $\emb_{i_2}$, which also implies that $v$ lies in the same face of $H$ in both embeddings.
         
    Thus $\embU$ is a simultaneous embedding of $\mc I$, which concludes the proof. \qedhere
\end{proof}

Combining \Cref{lm:twoConsEmb,lm:simEmb} with $l \geq |X| + 4$ concludes the proof of \Cref{lm:addBlob}.

\section{Hardness for Constant Pathwidth and Constant $\boldsymbol{k}$}
\label{section:pathwidth}

In this section, we derive~\cref{theorem:pathwidth_hardness}, i.e., \SEFE remains \NP-hard even for three input graphs with $\pw{G_\cup} \leq 5$.
Towards this, we first describe the \emph{two-star construction} (\cref{def:two-stars})---an \SEFE instance using 2 input graphs (which we will call red and blue, respectively) whose shared graph forms two disjoint stars.
In~\cref{lm:two-star-synchronization}, we essentially show that in any simultaneous embedding of the two-star construction, the circular order around a star center consists of ``regularly interleaving suborders''.
To derive hardness, we then add a third input graph (call it green) to the construction to further constrain the stars' rotation systems, which allows us to model \textsc{Bin Packing}~\cite{GareyJohnson}.

\begin{definition}[Two-star construction]
\label{def:two-stars}
Let $T$ be a finite nonempty set of \emph{types} and $L\geq 2$ an integer.
The \emph{two-star construction} $\mc S(T,L)=(G_1,G_2)$ is the \SEFE instance obtained as follows.
Construct $\shared$ by creating a star with center $v$ and attaching, for each $t\in T$ and $i\in\mathbb Z_L$, a leaf $t^i$, called the \emph{$i$th copy of type $t$}.
Add a disjoint copy of this star, where for each original vertex $u$, $\overline u$ denotes its copy.
Finally, for each $t\in T$ and $i\in\mathbb Z_L$, add the red exclusive edge $t^i\overline{t^i}$ to $G_1$ and the blue exclusive edge $t^i\overline{t^{i+1}}$ to $G_2$; see \Cref{fig:hardness} (left).
\end{definition}

\begin{restatable}\restateref{lm:two-star-synchronization}{lemma}{lemmaTwoStarSynchronization}
\label{lm:two-star-synchronization}
Let $T$ be a finite nonempty set and $L\geq 2$ an integer.
In every simultaneous embedding of $\mc S(T,L)$, there is $r\in\mathbb Z_{|T|L}$ such that, for every $t\in T$ and $i\in\mathbb Z_L$, the clockwise circular order around $v$ contains $t^{i+1}$ exactly $r$ positions after $t^i$.
Conversely, for any circular order around $v$ satisfying this property, one can realize a corresponding simultaneous embedding of $\mc S(T,L)$.
\end{restatable}
\begin{prooflater}{plemmaTwoStarSynchronization}
We use the following well-known fact.
If a graph $H$ consists of internally disjoint paths between two vertices $u$ and $w$, then a rotation system of $H$ is planar if and only if the clockwise order of the paths around $w$ is the reverse of their clockwise order around $u$.
Indeed, the bridges of a cycle formed by two of these paths are precisely the remaining paths, and the two fixtures of such a path lie on the same side of the cycle if and only if the three paths appear in opposite clockwise orders around $u$ and $w$; the fact thus follows from \Cref{lm:planarityCharacterization}.

Observe that $G_1$ consists of the internally disjoint paths $v\,t^i\,\overline{t^i}\,\overline v$ and $G_2$ consists of the internally disjoint paths $v\,t^i\,\overline{t^{i+1}}\,\overline v$ for $t \in T$ and $i \in \mathbb Z_L$.
For a circular order $\sigma$ of the leaves of $v$, let $\sigma_1$ and $\sigma_2$ denote the circular orders of the leaves of $\overline v$ obtained from $\sigma$ by replacing each $t^i$ with $\overline{t^i}$ and with $\overline{t^{i+1}}$, respectively.
By the above fact, the rotation system of $G_1$ with clockwise orders $\sigma$ around $v$ and $\overline \sigma$ around $\overline v$ is planar if and only if $\overline\sigma$ is the reverse of $\sigma_1$, and analogously for $G_2$ and $\sigma_2$.
Since $G_1$ and $G_2$ are connected and the shared graph consists of two stars, planar embeddings of $G_1$ and $G_2$ form a simultaneous embedding if and only if they have the same clockwise orders around $v$ and $\overline v$.
Hence $\sigma$ is the clockwise order around $v$ in a simultaneous embedding of $\mc S(T, L)$ if and only if $\sigma_1 = \sigma_2$.
As $\sigma_2$ is obtained from $\sigma_1$ by replacing each $\overline{t^i}$ with $\overline{t^{i+1}}$, this is the case if and only if $\sigma$ is invariant under the shift $t^i \mapsto t^{i+1}$, i.e., if the shift maps consecutive leaves of $\sigma$ to consecutive leaves.
This in turn is equivalent to the shift moving every leaf by the same number $r$ of positions, which concludes the proof.
\end{prooflater}

We reduce from the strongly \NP-hard \textsc{Bin Packing} problem~\cite{GareyJohnson}, which asks whether, given a finite set $U$ of items with sizes $s(u) \in \mathbb{Z}^+$ for each $u \in U$, a bin capacity $B \in \mathbb{Z}^+$, and an integer $K > 0$, the set $U$ can be partitioned into disjoint subsets $U_1, U_2, \ldots, U_K$ such that $\sum_{u \in U_i} s(u) \le B$ for all $i \in [K]$.
Fix such an instance with $U = \{u_1, \dots, u_n\}$.
Without loss of generality, we assume that $s(u) \leq B$ for all $u \in U$ and that $K \geq 3$.
We construct a corresponding instance of \SEFE $(G_1,G_2,G_3)$ as follows.
Let $T=\{R,t_1,\ldots,t_n\}$, where $R$ is the \emph{reference type} and the distinct types $t_j$ represent the items $u_j$.
Further, let $(G_1,G_2) = \mc S(T,K(B+1))$. To construct $G_3$, we begin with $G_1 \cap G_2$, and
subsequently add the \emph{bin cycle} $C=R^1\,R^{B+2}\,R^{2(B+1)+1}\,\cdots\,R^{(K-1)(B+1)+1}\,R^1$
and, for each item $u_j$, the \emph{item path} $P_j=t_j^1\,t_j^2\,\cdots\,t_j^{s(u_j)+1}$; see \Cref{fig:hardness} (right).

Finally, \cref{theorem:pathwidth_hardness} follows from the following lemma.

\begin{figure}[t]
    \centering
    \includegraphics[scale=0.9]{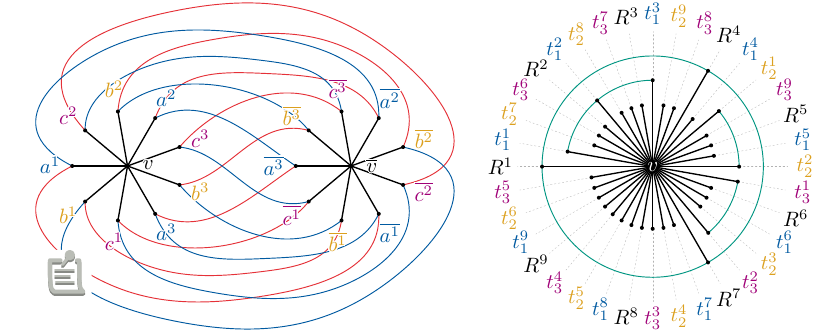}
\caption{   Left: The two-star construction for
  $T=\{\textcolor[rgb]{0,0.353,0.631}{a},
        \textcolor[rgb]{0.862,0.627,0.117}{b},
        \textcolor[rgb]{0.627,0,0.47}{c}\}$ and $L=3$, together with a simultaneous embedding; here, $r=3$ in \Cref{lm:two-star-synchronization}.
  Right: A resulting \SEFE instance for a \textsc{Bin Packing} instance
  with bin capacity $B=2$, $K=3$ bins, and item sizes
  $s(\textcolor[rgb]{0,0.353,0.631}{u_1})=2$,
  $s(\textcolor[rgb]{0.862,0.627,0.117}{u_2})=1$, and
  $s(\textcolor[rgb]{0.627,0,0.47}{u_3})=1$.
  The depicted order around $v$ (with $r=4$) corresponds to the solution
  $U_1=\{\textcolor[rgb]{0,0.353,0.631}{u_1}\}$,
  $U_2=\{\textcolor[rgb]{0.862,0.627,0.117}{u_2}, \textcolor[rgb]{0.627,0,0.47}{u_3}\}$, and $U_3=\emptyset$.
  The second star and red and blue edges are omitted to avoid clutter. }
    \label{fig:hardness}
\end{figure}

\begin{restatable}\restateref{lm:sefe-to-packing}{lemma}{lemmaSefeToPacking}
\label{lm:sefe-to-packing}
The \textsc{Bin Packing} instance is positive if and only if $(G_1,G_2,G_3)$ admits a simultaneous embedding. Moreover, $(G_1,G_2,G_3)$ can be constructed in polynomial time and the pathwidth of $G_\cup$ is at most $5$.
\end{restatable}
\begin{prooflater}{plemmaSefeToPacking}
Let $m \coloneqq |T| = n+1$ and $N \coloneqq mL$.
We call the vertices of $C$ \emph{bin vertices} and the edges of item paths \emph{item edges}.

Since \textsc{Bin Packing} is strongly \NP-hard, we can assume that all numbers are polynomially bounded; hence $\GU$, which has $2N+2$ vertices, can be constructed in polynomial time.
Deleting $v$, $\overline v$, and $t^1$ for every $t \in T$ from $\GU$ leaves, for each $t$, the path $\overline{t^2}\,t^2\,\overline{t^3} \cdots t^L\,\overline{t^1}$ together with green edges that close edge-disjoint cycles with it, which has pathwidth at most~$2$.
As the vertices $t^1$ lie in distinct components of $\GU - \{v, \overline v\}$, we obtain $\pw{\GU} \leq 2 + 1 + 2 = 5$.

Now let $U_1, \dots, U_K$ be a solution of the \textsc{Bin Packing} instance.
For $u_j \in U_w$, let $o_j \coloneqq (w - 1)(B+1) + \sum_{j' < j,\, u_{j'} \in U_w} s(u_{j'})$ and let $\sigma$ be the circular order obtained by concatenating the blocks $R^\ell, t_1^{\ell - o_1}, \dots, t_n^{\ell - o_n}$ for $\ell = 1, \dots, L$; see \Cref{fig:hardness} (right).
As $t^{i+1}$ lies $m$ positions after $t^i$ in $\sigma$, \Cref{lm:two-star-synchronization} yields a simultaneous embedding of $(G_1, G_2)$ with clockwise order $\sigma$ around $v$.
Since the star of $\overline v$ is a connected component of $G_3$, it remains to embed the component $H$ of $G_3$ containing $v$ with clockwise order $\sigma$ around $v$.
For $u_j \in U_w$, the vertices of $P_j$ lie in the blocks $o_j + 1, \dots, o_j + s(u_j) + 1$, i.e., in the order of $P_j$ between the consecutive bin vertices $R^{1 + (w-1)(B+1)}$ and $R^{1 + w(B+1)}$ (as $\sum_{u \in U_w} s(u) \leq B$) and after all item paths $P_{j'}$ with $j' < j$ and $u_{j'} \in U_w$.
Hence no two green edges have alternating endpoints in $\sigma$.
Adding the cycle $Z$ through the leaves of $v$ in the order $\sigma$ thus yields a wheel whose planar embedding with clockwise order $\sigma$ around $v$ can be extended by embedding the green edges as chords into the outer face of $Z$; removing $Z$ yields the desired embedding of $H$.

Conversely, let $(\emb_1, \emb_2, \emb_3)$ be a simultaneous embedding, let $\sigma$ be the clockwise order around $v$, and let $\mathrm{pos}(x) \in \mathbb Z_N$ denote the clockwise position of a leaf $x$ in $\sigma$.
By \Cref{lm:two-star-synchronization}, there is $r \in \mathbb Z_N$ with $\mathrm{pos}(t^{i+1}) = \mathrm{pos}(t^i) + r$ for all $t \in T$ and $i \in \mathbb Z_L$.
As the $L$ copies of each type $t$ have distinct positions $\mathrm{pos}(t^1) + jr$ for $j \in \mathbb Z_L$, the multiples of $r$ are precisely the multiples of $m$.
Hence the copies of $R$ occupy every $m$-th position and the bin vertices, which have positions $\mathrm{pos}(R^1) + w(B+1)r$ for $w \in \mathbb Z_K$, occupy every $(B+1)m$-th position, splitting the remaining leaves into $K$ \emph{bins} with $B$ copies of $R$ each.

Let $xx'$ be an item edge with $x = t_j^i$ and $x' = t_j^{i+1}$.
As the cycle $v\,x\,x'$ is disjoint from $C$, all bin vertices lie on the same side of it by \Cref{lm:planarityCharacterization}, so $x$ and $x'$ lie in a common bin.
Hence each item path lies within a single bin and, treating positions within a bin as integers, $\mathrm{pos}(x') - \mathrm{pos}(x) = \tilde r$ for the unique integer $\tilde r \equiv r \pmod N$ with $|\tilde r| < N/2$, as bins have fewer than $(B+1)m \leq N/2$ leaves.
Since $\tilde r \neq 0$ is a multiple of $m$, the open interval $I_{xx'}$ between $\mathrm{pos}(x)$ and $\mathrm{pos}(x')$ contains a copy of $R$.
For an item edge $yy'$ of another item path, which is disjoint from $v\,x\,x'$, \Cref{lm:planarityCharacterization} implies that $y$ and $y'$ lie on the same side of $v\,x\,x'$.
Thus $I_{xx'}$ and $I_{yy'}$ do not cross and, having the same length, are not nested either, i.e., they are disjoint.
The same holds for distinct edges of the same item path, as its consecutive vertices are $\tilde r$ positions apart.
Hence, if $U_w$ denotes the set of items whose paths lie in the $w$-th bin, the $\sum_{u \in U_w} s(u)$ item edges in this bin have disjoint intervals that each contain one of its $B$ copies of $R$.
Therefore, $\sum_{u \in U_w} s(u) \leq B$ for all $w \in [K]$, which concludes the proof.
\end{prooflater}

\section{Concluding Remarks}
\label{section:conclusion}
Taken together, the results presented in this work bring us closer to a complete understanding of the boundaries of tractability for \ssefe. A natural open question arising from our work is whether our result can be made more ``algorithmically tame'' by avoiding the tower of exponents in the running time guarantees. More broadly, it would be interesting to settle the complexity of $2$-\sefe\ and to understand whether any of the recent results for \ssefe can be lifted to the non-sunflower setting.

\subparagraph*{Statement on the Use of AI Tools.} The authors acknowledge the use of AI tools (GPT-6 Astra and Claude Opus 5.5) for proofreading and correctness verification. Moreover, the idea underlying the hardness reduction (Theorem~\ref{theorem:pathwidth_hardness}) complementing our main result originated from a discussion with GPT-6 Astra; however, no such tools were used for the idea and writeup of the proof of the main result (Theorem~\ref{theorem:vertex_integrity_fpt}). The authors assume full responsibility for all claims in the article.

\bibliography{refs}

\onlyShort{
\newpage
\appendix

\crefalias{section}{appendix}
\crefalias{subsection}{appendix}

\section{Omitted Proofs from \Cref{section:prelims}}
\label{appendix:prelims}  

\lemmaPlanarityCharacterization*
\label{lm:planarityCharacterization*}
\plemmaPlanarityCharacterization

\lemmaMerges*
\label{lm:merges*}
\plemmaMerges

\lemmaExclConnected*
\label{lm:exclconnected*}
\plemmaExclConnected

\section{Omitted Proofs from \Cref{section:vertex-integrity}}
\label{appendix:vertex-integrity}  

\lemmaConnectedSubinstances*
\label{lm:connectedSubinstances*}
\plemmaConnectedSubinstances

\lemmaReducedEquivalent*
\label{lm:reducedEquivalent*}
\plemmaReducedEquivalent

\lemmaComputeReduced*
\label{lm:computeReduced*}
\plemmaComputeReduced

\lemmaTwoconsistent*
\label{lm:twoconsistent*}
\plemmaTwoconsistent

\lemmaAngleAssignment*
\label{lm:angleAssignment*}
\plemmaAngleAssignment

\lemmaTwoConsEmb*
\label{lm:twoConsEmb*}
\plemmaTwoConsEmb

\lemmaMConsistent*
\label{lm:m-consistent*}
\plemmaMConsistent

\claimSmallObstruction*
\label{cl:smallObstruction*}
\pclaimSmallObstruction

\section{Omitted Proofs from \Cref{section:pathwidth}}
\label{appendix:pathwidth}

\lemmaTwoStarSynchronization*
\label{lm:two-star-synchronization*}
\plemmaTwoStarSynchronization

\lemmaSefeToPacking*
\label{lm:sefe-to-packing*}
\plemmaSefeToPacking
}

\end{document}